\documentclass[11pt]{amsart}

\usepackage[letterpaper,asymmetric,left=1in,right=1in,top=1.2in,bottom=1.
25in,bindingoffset=0.0in]{geometry}

\usepackage{graphicx}
\usepackage{url}
\usepackage{bm}
\usepackage{rotate}
\usepackage{color}
\usepackage{amssymb}

\newtheorem{theorem}{Theorem}[section]
\newtheorem{lemma}[theorem]{Lemma}
\newtheorem{proposition}[theorem]{Proposition}

\newtheorem*{CPL}{Dujardin-Favre Classification of Passivity Locus \cite[Theorem 4]{dujardin}}

\newcommand{\beq}{\begin{equation}}
\newcommand{\eeq}{\end{equation}}
\newcommand{\beqs}{\begin{eqnarray}}
\newcommand{\eeqs}{\end{eqnarray}}

\begin{document}

\title{$q$-Plane Zeros of the Potts Partition Function  \\ on Diamond 
Hierarchical Graphs}

\begin{author}{Shu-Chiuan Chang}
\email{scchang@mail.ncku.edu.tw}
\address{Physics Department,
National Cheng Kung University, Tainan 70101, Taiwan}
\end{author}

\begin{author}{Roland K. W. Roeder}
\email{rroeder@math.iupui.edu}
\address{ %
IUPUI Department of Mathematical Sciences\\
LD Building, Room 224Q\\
402 North Blackford Street\\
Indianapolis, Indiana 46202-3267\\
 United States }
\end{author}

\begin{author}{Robert Shrock}
\email{robert.shrock@stonybrook.edu}
\address{C. N. Yang Institute for Theoretical Physics and \\
Department of Physics and Astronomy \\
Stony Brook University, Stony Brook, NY 11794 \\
United States }
\end{author}

\begin{abstract}

  We report exact results concerning the zeros of the partition function of the
  Potts model in the complex $q$ plane, as a function of a temperature-like
  Boltzmann variable $v$, for the $m$'th iterate graphs $D_m$ of the Diamond
  Hierarchical Lattice (DHL), including the limit $m \to \infty$. In this limit
  we denote the continuous accumulation locus of zeros in the $q$ planes at
  fixed $v = v_0$ as ${\mathcal B}_q(v_0)$.  We apply theorems from complex
  dynamics to establish properties of ${\mathcal B}_q(v_0)$.  For $v=-1$ (the
  zero-temperature Potts antiferromagnet, or equivalently, chromatic
  polynomial), we prove that ${\mathcal B}_q(-1)$ crosses the real-$q$ axis at
  (i) a minimal point $q=0$, (ii) a maximal point $q=3$ (iii) $q=32/27$, (iv) a
   cubic root that we give, with the value $q = q_1 = 1.6388969..$, and 
  (v) an infinite number of points smaller than $q_1$, converging to $32/27$ 
  from above.  Similar results hold for ${\mathcal B}_q(v_0)$ for any 
  $-1 < v < 0$ (Potts antiferromagnet at nonzero
  temperature).  The locus ${\mathcal B}_q(v_0)$ crosses the real-$q$ axis at
  only two points for any $v > 0$ (Potts ferromagnet).  We also provide
  computer-generated plots of ${\mathcal B}_q(v_0)$ at various values of $v_0$
  in both the antiferromagnetic and ferromagnetic regimes and compare them to
  numerically computed zeros of $Z(D_4,q,v_0)$.

\vspace{0.1in}

\end{abstract}

\maketitle

\section{Introduction} 
\label{intro}

We derive some exact results concerning the zeros in the complex $q$ plane of
the partition function, $Z(D_m,q,v)$, for the $q$-state Potts model on Diamond
Hierarchical Graphs $D_m$ at various fixed values of a temperature-like
Boltzmann variable $v$.  We also derive exact results concerning the continuous
accumulation set ${\mathcal B}_q(v_0)$ of these zeros on the limiting Diamond
Hierarchical Lattice $D_\infty$, again at various fixed values of $v = v_0$,
and we present computer-generated images of these loci.

The Diamond Hierarchical Graphs $D_m$ are defined by starting with a graph
$D_0$ consisting of two vertices (sites) and an edge (bond) joining them. The
iterative graphical transformation replaces this single edge by four edges and
two additional vertices, as shown in Fig. \ref{dhlfigure}, yielding the next
iterate, $D_1$ (which has the appearance of a diamond, whence the name).  Fig.
\ref{dhlfigure} also shows the next iterate, $D_2$. We shall use the term
Diamond Hierarchical Lattice (DHL) to refer to the
(formal) limit $\lim_{m \to \infty} D_m \equiv D_\infty$.
It is a self-similar, fractal object.

\begin{figure}
\begin{center}
\scalebox{1.2}{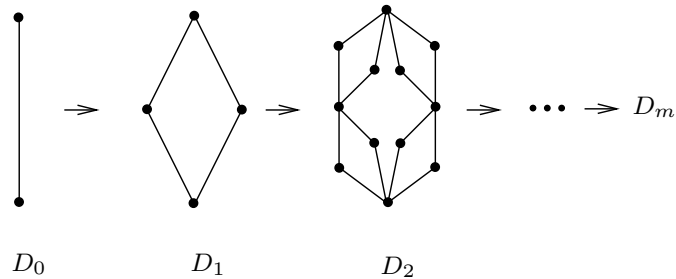}
\end{center}
\caption{Diamond hierarchical graphs $D_m$ for $m=0, \ 1, \ 2$.}
\label{dhlfigure}
\end{figure}

We recall the procedure for calculating the Hausdorff dimension $d_H$ of a
hierarchical lattice $G_\infty$.  If the renormalization-group (RG)
transformation reduces the length of each edge by a blocking factor of $b$ and
gives rise to $N$ copies of the original graph, then $N = b^{d_H}$, so
\cite{gk,gam} $d_H = \ln(N)/\ln(b)$.  In the case of the iteration procedure
for DHL, one has $b=2$ and $N=4$, yielding the well-known result that
\beq
d_{H}(D_\infty)=2 \ .
\label{dhldim}
\eeq For this reason, we interpret the Diamond Hierarchical Lattice $D_\infty$
as being two-dimensional.  (See also \cite[Appendix E.3]{blr} for an
interpretation of the Diamond Hierarchical Lattice as an anisotropic version of
the $\mathbb{Z}^2$ lattice.)

The $q$-state Potts model has been of longstanding interest in the
area of phase transitions and critical phenomena.  On a graph $G$, the
partition function of this model, denoted $Z(G,q,v)$, is a polynomial
in two variables, $q$ and 
\beq
v \equiv y-1 \ .
\label{vy}
\eeq
In the original statistical physics formulation, $q$ is a positive integer
specifying the number of possible values of a classical spin defined at a given
site of a lattice, $\sigma(i) \in \{1,...,q\}$, and $y$ is a non-negative
temperature-like Boltzmann variable.  (Throughout this paper we will primarily
use the variable $v$ because certain expressions are simpler in $v$ rather than
$y$.)  As is evident from the expression (\ref{cluster}) given below for this
partition function, it is a polynomial in both $q$ and $v$, and one can
generalize both of these variables to complex values. Indeed, this
generalization is necessary when analyzing the zeros of the partition function
in the $q$ plane for a given value of $v$ and in the $v$ plane for a given
value of $q$.

Part of the interest in the Potts model partition function stems from the fact
that it is equivalent to a function of central importance in mathematical graph
theory, namely the Tutte polynomial, $T(G,x,y)$ (see Eq.\ (\ref{zt})
below). For some basic background on graph theory and the Tutte and chromatic
polynomials, see, e.g., \cite{biggs}-\cite{dong}.

On a family of $n$-vertex lattice graphs, as $n \to \infty$, an infinite subset
of the zeros of $Z(G,q,v)$ merge to form certain continuous loci.  In this $n
\to \infty$ limit, we denote the continuous accumulation locus of zeros of
$Z(G,q,v)$
\begin{itemize}
\item[(i)] in the complex $q$ plane, for a given $v = v_0$, as ${\mathcal B}_q(v_0)$, and
\item[(ii)] in the complex $v$ 
plane, for a given $q = q_0$, as ${\mathcal B}_v(q_0)$.
\end{itemize}
In this paper we will primarily be interested in the $q$-plane loci ${\mathcal
  B}_q(v_0)$, however it will occasionally useful to relate them to the
$v$-plane loci ${\mathcal B}_v(q_0)$ and to discuss similarities and
differences between these loci.

Although no exact closed-form expression for $Z(G,q,v)$ with general $q$ and
$v$, or for the corresponding dimensionless reduced free energy has been
obtained on (the thermodynamic limit of) any regular lattice graph $G$ of
spatial dimension $d \ge 2$, it has been possible to characterize the
renormalization group (RG) action for the model exactly on certain hierarchical
lattice graphs, including the Diamond Hierarchical Graphs.
By performing a sum over spins at each
iterative step, one can construct an exact RG transformation
relating $Z(D_{m+1},q,v)$ to $Z(D_m,q,v')$, where $v'$ is related to
$v$ according to a function $v'=F_q(v)$ or equivalently, $y'=r_q(y)$,
(see Eqs.\ (\ref{vp}) and (\ref{yp}) below).  This result
follows because of the self-similarity of the Diamond Hierarchical Lattice. The properties
of this model in the $m \to \infty$ limit are then determined by the
properties of the iterated function $F_q$ or equivalently $r_q$.  The
properties of iterated analytic functions have been of considerable importance
in mathematics (e.g., \cite{beardon}-\cite{milnor} and physics
\cite{stanley_taylor}-\cite{bunde_havlin}). As will be clear, there are also
interesting connections with complex analysis (see, e.g., \cite{remmert}).

There have been many studies of spin models on hierarchical lattices, including
\cite{berker}-\cite{yangzeng}, which primarily analyze the zeros of
$Z(D_m,q,v)$ in the complex plane of the temperature-like Boltzmann variable
$v$.  (A notable exception is \cite{blr} where the Lee-Yang (complex magnetic
field) and Lee-Yang-Fisher (complex magnetic field and complex $v$
simultaneously) zeros are studied for the Diamond Hierarchical Lattice.)  It
was natural for these previous works to focus on the $v$-plane zeros and their
continuous accumulation set as $m \to \infty$, ${\mathcal B}_v(q_0)$, for a
given $q=q_0$, because it is directly related to the iteration of the RG
transformation $F_{q_0}(v)$ at fixed value of the parameter $q_0$.  Indeed, (in
most settings) ${\mathcal B}_v(q_0)$ is the Julia set in the $v$ plane for the
mapping $F_{q_0}(v)$.

Considerably less attention has been paid to the zeros of $Z(D_m,q,v)$ in the
$q$ plane and their continuous accumulation set as $m \to \infty$, ${\mathcal
B}_q(v_0)$, at fixed values of $v=v_0$. Rather than being the Julia set of a
rational mapping, ${\mathcal B}_q(v_0)$ is related to the parameter dependence
of the iterates of $F_q(v)$ for the fixed choice of initial condition $v=v_0$.
(This will be elaborated in Section \ref{complex_dynamics}.) We have noted
above the study of the zeros in the $q$ plane for the Sierpinski gasket
\cite{sg}.   In the case of the Diamond Hierarchical Lattice
the locus ${\mathcal B}_q(-1)$ has been recently studied in \cite{WQYQG},
\cite{yangzeng}, and~\cite{rr}.  Wang, Qui, Yin, Qiao, and Gao  \cite{WQYQG}
and Yang and Zeng \cite{yangzeng} proved that the bifurcation locus
${\mathcal M}$ for the renormalization mapping $F_q(v)$ given in Eqn.\
(\ref{zdhl}) is connected.  In Ref. \cite{rr}, Chio and Roeder use techniques
from complex dynamics to show that ${\mathcal M} \subset {\mathcal B}_q(-1)$
for the DHL and, in particular to prove for the DHL that the Hausdorff
dimension of ${\mathcal B}_q(-1)$ is 2.  The paper \cite{rr} also provides a
quantitative description of the limiting behavior of the chromatic zeros in
terms of measure theory.

In this paper we will study properties of the loci ${\mathcal B}_q(v_0)$ at
various choices of $v_0$, beyond the case of the chromatic zeros $v_0 = -1$.
Using techniques from complex dynamics similar to those in \cite{rr}, we will
make computer images (see Figures \ref{bq_vm1_figure} - \ref{bq_v99_figure}) of
these loci, which we relate to numerical computations of the 172 zeros of
$Z(D_4,q,v_0)$.  We will also rigorously determine properties of the
intersection between ${\mathcal B}_q(v_0)$ and the real $q$-axis.  The latter
results are new even in the case of the chromatic zeros $v=-1$. We will also
make use of results from statistical mechanics to gain further insight into
the properties of ${\mathcal B}_q(v_0)$. 

This paper is organized as follows. In Sections
\ref{background}-\ref{transformation} we review some relevant background on the
Potts model, the family of Diamond Hierarchical Graphs, $D_m$, and the
iterative RG transformation $F_q(v)$ that relates $Z(D_{m+1},q,v)$ to
$Z(D_m,q,v')$.  In Section \ref{complex_dynamics} we present some necessary
background in complex dynamics and use it to relate the locus ${\mathcal
  B}_q(v_0)$ to the ``active parameters'' $q$ for the RG transformation.  In
Sections \ref{chromaticzeros} and \ref{qzerosafm} we present our results on
zeros of the partition function in the $q$ plane for the Potts antiferromagnet
at zero and finite temperature, respectively.  We also state Theorem
\ref{THM:AFM} describing the intersections of ${\mathcal B}_q(v_0)$ with the
real $q$-axis in this regime ($-1 \leq v_0 < 0$).  In Section \ref{qzerosfm}
presents results on the zeros in the $q$ plane for the Potts ferromagnet,
including the statement of Theorem \ref{THM:FM} describing the intersections of
${\mathcal B}_q(v_0)$ with the real $q$-axis in this regime ($v_0 >
0$). Section \ref{SEC:PROOFS} is devoted to proofs of Theorems \ref{THM:AFM}
and \ref{THM:FM}.  Section \ref{yzeros} contains our results on partition
function zeros in the $v$ plane for various values of $q$. Our conclusions are
summarized in Section~\ref{conclusions}, and some auxiliary information is
given in Appendix \ref{q0_appendix}.


\section{Background from Graph Theory and Statistical Physics} 
\label{background}

In this section we discuss some relevant background from graph theory
and statistical physics.  A graph $G=(V,E)$ is defined by its set
$V$ of vertices (= sites) and its set $E$ of edges (= bonds).  We denote
$n=n(G)=|V|$ and $e(G)=|E|$ as the number of vertices and edges of
$G$. At temperature $T$, the partition function of the $q$-state Potts
model is given by 
\beq Z = \sum_{ \sigma } e^{-\beta {\mathcal H}(\sigma)} \ ,
\eeq
with the Hamiltonian
\beq
{\mathcal H}(\sigma) = -J \sum_{e_{ij}} \delta_{\sigma(i), \sigma(j)} \ .
\label{ham}
\eeq
Here, the sum is taken over all edges $e_{ij}$ of $G$, with $i$ and $j$
labeling vertices of $G$; $\sigma: V \rightarrow \{1,...,q\}$ is an assignment
of classical spins to the vertices; $\beta = (k_BT)^{-1}$; $J$ is the spin-spin
interaction constant; and $k_B$ is the Boltzmann constant \cite{wurev}. 
Further, $\delta_{r,s}$ is the Kronecker delta function.  We define the 
notation
\beq
K = \beta J \ , \quad y = e^K \ , \quad \mbox{and} \quad v = y-1.
\label{kdef}
\eeq
The signs of $J$ favoring ferromagnetic (FM) and antiferromagnetic (AFM) spin
configurations are $J > 0$ and $J < 0$, respectively.  Hence, the physical
ranges of $v$ are $v \ge 0$ for the Potts ferromagnet (FM) and $-1 \le v \le 0$
for the Potts antiferromagnet (AFM).  The partition function for the $q$-state
Potts model can equivalently be written as
\beq
Z = \sum_{\sigma} \prod_{e_{ij}} (1 + v\delta_{\sigma(i) \sigma(j)}) \ . 
\label{zv}
\eeq
Thus $Z$ is invariant under a global symmetry that acts on the spins,
namely for any permutation $\pi_q$ of $\{1,\ldots,q\}$ we can apply the mapping $\sigma(i) \to \pi_q(\sigma(i))$ to
the spin at each site $i$, leaving $Z$ unchanged.
At high temperatures, this symmetry is realized explicitly in the
physical states, while in the $n \to \infty$ (thermodynamic) limit on
a lattice graph with dimensionality greater than a lower critical
dimensionality, it can be broken spontaneously with the presence of a
nonzero long-range ordering of the spins.  This ordering is
ferromagnetic or antiferromagnetic, depending on the sign of $J$.

A spanning subgraph of $G$ is $G' = (V,E')$ with $E' \subseteq E$. The number
of connected components of $G'$ is denoted $k(G')$.  The partition function of
the Potts model can equivalently be expressed in a purely graph-theoretic manner 
as the sum over spanning subgraphs \cite{fk}
\beq
Z(G,q,v) = \sum_{G' \subseteq G} q^{k(G')} \, v^{e(G')}  \ .
\label{cluster}
\eeq
Eq.\ (\ref{cluster}) shows that the partition function $Z(G,q,v)$ is a
polynomial in $q$ and $v$ with positive integer coefficients for each
nonzero term. As is evident from Eq.\ (\ref{cluster}),
$Z(G,q,v)$ has degree $n(G)$ in $q$ and $e(G)$ in $v$, or equivalently, in $y$.
Furthermore, Eq.\ (\ref{cluster})
allows one to generalize the parameter $q$ beyond the positive
integers, ${\mathbb Z}_+$. In particular, for the ferromagnetic case
$v > 0$, Eq.\  (\ref{cluster}) allows one to generalize $q$ from the
positive integers to the positive real numbers while keeping a Gibbs
measure, i.e., keeping $Z(G,q,v) > 0$. (This is not, in general,
possible for the antiferromagnetic case, except when $q$ is an
integer, so one can revert to the Hamiltonian formulation in
Eq.\ (\ref{ham}) and (\ref{zv}), since $v$ is negative, so $Z(G,q,v)$
contains terms of both signs.)  More generally, Eq.\ (\ref{cluster})
allows one to generalize both $q$ and $v$ from their physical ranges
to complex values, as is necessary in order to analyze the zeros of
$Z(G,q,v)$ in $q$ for fixed $v$ and the zeros of $Z(G,q,v)$ in $v$ for
fixed $q$.  Since the coefficients in $Z(G,q,v)$ are real (actually in
${\mathbb Z}_+$, but all we use here is the reality), it follows that
for real $v$, the zeros of $Z(G,q,v)$ in the $q$ plane are invariant
under complex conjugation and for real $q$, the zeros of $Z(G,q,v)$ in
the $v$ plane plane are invariant under complex conjugation.
Since $k(G') \ge 1$ for all $G'$, it follows that $Z(G,q,v)$ always
contains an overall factor of $q$.  We can thus define a reduced
partition function
\beq
Z_r(G,q,v) = q^{-1} Z(G,q,v) \ , 
\label{zr}
\eeq
which is also a polynomial in $q$ and $v$

Let us denote $G_\infty$ as the formal limit, as $n \to \infty$ on a
family of graphs $G_n$ (here, $G_n = D_m$). In this limit, the dimensionless,
reduced free energy, per vertex, is defined as
\beq
f(G_\infty,q,v) = \lim_{n \to \infty} n^{-1} \ln [Z(G_n,q,v)] \ .
\label{fz}
\eeq
(The actual free energy is equal to $ -k_B T f$.) 
For the Potts antiferromagnet, $T \to 0$ means $K \to -\infty$ and
thus $v \to -1$. As is clear from Eq.\ (\ref{zv}), the only spin
configurations that contribute to $Z(G,q,v)$ in this limit are those
for which the spins on adjacent vertices are different.  Hence,
\beq
P(G,q) = Z(G,q,-1) \ , 
\label{pz}
\eeq
where $P(G,q)$ is the chromatic polynomial, which, for $q \in {\mathbb
  Z}_+$, counts the number of ways of assigning $q$ colors to the
vertices of $G$ subject to the condition that no two adjacent vertices
have the same color (called proper $q$-colorings of $G$).
The minimum integer $q$ that allows a proper
$q$-coloring of $G$ is the chromatic number, $\chi(G)$.
Since $D_m$ is bipartite, $\chi(D_m)=2$. Since $P(G,q)$
always contains a factor of $q$, we also define
\beq
P_r(G,q) = q^{-1}P(G,q) = Z_r(G,q,-1) \ .
\label{pr}
\eeq

Besides its intrinsic interest in mathematical graph theory, the
chromatic polynomial is important for physics because of its
connection with ground-state entropy.  For a family $G_m$, in the
limit $m~\to~\infty$ (and hence $n \to \infty$), the configurational
degeneracy per site (vertex) of the Potts antiferromagnet is
\beq
W(G_\infty,q) = \lim_{n \to \infty} P(G_m,q)^{1/n} \ . 
\label{w}
\eeq
For real $q < \chi(G_m)$, $P(G_m,q)$ can be negative; in this
case, since there is no obvious choice for which of the $n$ roots
of $(-1)$ to pick, one can only determine $|W(G_m,q)|$ \cite{w}. 
For a set of special $q$ values, $\{ q_s \}$, the limits $n \to
\infty$ and $q \to q_s$ do not, in general, commute for
$P(G,q)^{1/n}$, so one must specify the order of limits in defining
$W(G_\infty,q)$ \cite{w,a}.  The set $\{ q_s \}$ depends on the family
$G_m$ but usually includes $q=0$ and $q=1$. Here we define the order
of limits as $q \to q_s$ first and then $n \to \infty$.  For a wide
class of $G_m$ families, if $q$ is sufficiently large, then the number
of proper $q$-colorings of $G_m$ grows exponentially with $m$ and $n$,
so that $W(G_\infty,q) > 1$ and hence the Potts AFM has nonzero
ground-state entropy per vertex on
$G_\infty$, $S(G_\infty,q) = k_B \ln[W(G_\infty,q)] \ne 0$.

The Tutte polynomial, denoted $T(G,x,y)$, of a graph $G$ is defined by
\beq
T(G,x,y) = \sum_{G' \subseteq G} (x-1)^{k(G')-k(G)} (y-1)^{c(G')} \ ,
\label{t}
\eeq
where, as above, $k(G')$ denotes the number of connected components of the
spanning subgraph $G'$, and $c(G')$ denotes the number of linearly independent
circuits on $G'$, given by $c(G')=e(G')+k(G')-n(G')$.  With $y=e^K$, as defined in
Eq.\ (\ref{kdef}) and
\beq
x = 1 + \frac{q}{v} \ ,
\label{x}
\eeq
it follows that
\beq
Z(G,q,v) = (x-1)^{k(G)}(y-1)^{n(G)}T(G,x,y) \ .
\label{zt}
\eeq
Thus, the partition function of the Potts model is equivalent, up to the
indicated prefactor, to the Tutte polynomial on a given graph $G$.

Zeros of $Z(G,q,v)$ in $q$ for a given $v$ and in $v$ for a given $q$ are of
interest partly because for many families of graphs, such as strips of regular
lattices, in the $m \to \infty$ limit, an infinite subset of these respective
zeros typically merge to form certain continuous loci.  As stated above, for a
one-parameter family of graphs $G_m$, we define the locus ${\mathcal B}_q(v_0)$
as the continuous accumulation set of zeros of $Z(G_m,q,v_0)$ in the complex
$q$ plane as $m \to \infty$ for a fixed complex value $v=v_0$. (There may also
be discrete zeros that do not lie on this locus.)  Similarly, we define the
locus ${\mathcal B}_v(q_0)$, or equivalently, ${\mathcal B}_y(q_0)$, as the
continuous accumulation set of zeros of $Z(G_m,q_0,v)$ in the complex plane of
the temperature-dependent variable $v$, or equivalently, $y$, as $m \to \infty$
for a fixed complex value $q = q_0$ \cite{a}.

For infinite-length, finite-width strips of regular lattices, and also chain
graphs, ${\mathcal B}_q(v_0)$ is generically comprised of real algebraic
curves, including possible line segments \cite{w}-\cite{SS2} (for a review and
further references, see \cite{rev}). The underlying
reason for this is that $P(G,q)$, and more generally, $Z(G,q,v)$, for these
classes of graphs consist of a sum of $m$'th powers of certain algebraic
functions, denoted generically as $\lambda_j$, where $m$ is the length of the
strip, and the continuous loci ${\mathcal B}_q(v_0)$ occur at values of $q$
where there are two or more $\lambda_j$ functions that are largest in magnitude
and degenerate in magnitude. An early mathematical analysis of this sort of
behavior was given in \cite{bkw2,bkw1}. The Tutte polynomials for such strip
graphs satisfy certain recursion relations \cite{bds}. The loci ${\mathcal
  B}_q(v_0)$ may be connected, as, e.g., for the $m \to \infty$ limit of the
circuit graph and strips of the square, triangular, and honeycomb lattices with
periodic or twisted periodic (M\"obius) longitudinal boundary conditions, and
corresponding strips with toroidal or Klein-bottle boundary conditions
\cite{baxter87},\cite{w}-\cite{sdg}.  These generically separate the $q$ plane
into different regions. For infinite-length limits of strip graphs and chain
graphs with periodic longitudinal boundary conditions, ${\mathcal B}_q(v_0)$
crosses the real $q$ axis at a maximal point denoted $q_c(G_\infty)$, as well
as at one or more other points \cite{w}. If $q_c(G_\infty) \in {\mathbb Z}_+$,
this corresponds to the property that for this value of $q$, ${\mathcal
  B}_v(q)$ passes through $v=-1$, signifying that the Potts antiferromagnet
with $q=q_c(G_\infty)$ has a $T=0$ critical point on $G_\infty$.  For a wide
class of families of graphs, ${\mathcal B}_q(v)$ also crosses the real $q$ axis
on the left at $q=0$, but there are self-dual families where this left crossing
is shifted to $q=1$ \cite{sdg}. In contrast, for infinite-length limits of
lattice strip graphs with free longitudinal and transverse boundary conditions,
${\mathcal B}_q(v_0)$ consists generically as (complex-conjugate pairs of) arcs
and possible real line segments that do not separate the $q$ plane into
different regions \cite{strip}.  In all of these cases, the
Hausdorff dimension of ${\mathcal B}_q(v_0)$ is 1.  The result proved in
\cite{rr} that the Hausdorff dimension of ${\mathcal B}_q(-1)$ is 2 for the
Diamond Hierarchical Lattice, thus provides an interesting contrast with the
behavior of ${\mathcal B}_q(-1)$ for these other infinite-$n$ limits of
families of graphs.

The locus ${\mathcal B}_v(q_0)$ is commonly also comprised of curves and
possible line segments in the $v$ plane.  Although this is common
behavior, it is not always true; a counterexample to this was found
for the $(4 \cdot 8 \cdot 8)$ Archimedean lattice, where these zeros
form areas rather than curves even in the case of isotropic spin-spin
interaction constants considered here \cite{cmo}.  These areas reduce
to points where ${\mathcal B}_v(q_0)$ crosses the real axis.

Note that the property that ${\mathcal B}_q(v_0)$ crosses the real $q$ axis at
a point $q_0$ does not imply that $Z(G,q,v_0)$ vanishes at the point $q_0$, and
similarly, the property that ${\mathcal B}_v(q_0)$ crosses the real $v$ axis at
a point $v_0$ does not imply that $Z(G,q_0,v)$ vanishes at that point
$v_0$. Indeed, from Eq.\ (\ref{cluster}), it is evident that $Z(G,q,v)$ is a
polynomial in $q$ and $v$ with positive coefficients, and hence for fixed
positive $v$, $Z(G,q,v)$ has no zeros for any real positive $q$, and for fixed
positive $q$, $Z(G,q,v)$ has no zeros for any real positive $v$.  The property
that the continuous accumulation locus of the chromatic zeros ${\mathcal
  B}_q(-1)$ crosses the real $q$ axis at a point $q_0$ also does not imply that
$P(G,q_0)=0$, although here the argument is more subtle, since $P(G,q)$ has
terms that alternate in sign with descending powers of $q$.  The precise
meaning of the statement that, for a given $v_0$, the continuous locus
${\mathcal B}_q(v_0)$ crosses the real $q$ axis at a point $q_0$ is that in the
limit $n(G) \to \infty$, the zeros of $Z(G_m,q,v_0)$ approach arbitrarily close
to $q_0$.  This type of behavior is familiar from statistical physics.  For
example, for the $q$-state Potts model on the square lattice with integral $q
\ge 2$ one has that ${\mathcal B}_v(q)$ crosses the real $v$ axis at
$v_c=\sqrt{q}$.  This critical point separates the paramagnetic phase with $0
\le v \le v_c$ with explicit $S_q$ symmetry from the ferromagnetically ordered
phase with $v > v_c$, in which the $S_q$ symmetry is spontaneously broken 
(e.g., \cite{wurev}).


\section{Diamond Hierarchical Graphs $D_m$ and Diamond Hierarchical Lattice $D_\infty$} 
\label{dh}

In this section we discuss further details of Diamond Hierarchical
Lattice graphs $D_m$ and the limit $m \to \infty$.  We have discussed
above how one defines $D_m$ iteratively, starting with $D_0=T_2$, the
tree graph with two vertices. The numbers of
vertices and edges on $D_m$ are
\beq
n(D_m) = \frac{2(4^m + 2)}{3} \qquad \mbox{and} \qquad e(D_m) = 4^m \ .
\eeq

The degree, $\Delta$, of a vertex on a graph $G$ is defined as the
number of edges that connect to it.  (The word ``degree'' is thus used
in two different ways, here, but this should not cause any confusion.)
A $\Delta$-regular graph is a graph all of whose vertices have the
same degree. Although $D_0$ and $D_1$ are $\Delta$-regular graphs with
$\Delta=1$ and $\Delta=2$, respectively, the $D_m$ graphs with $m \ge
2$ are not $\Delta$-regular, but instead have vertices with degrees
ranging from 2 to $2^m$. For an arbitrary graph $G$, the average
(effective) vertex degree is
\beq
\Delta_{eff}(G) = \frac{2e(G)}{n(G)} \ . 
\label{delta_eff}
\eeq
For the Diamond Hierarchical Graphs
\beq
\Delta_{eff}(D_m) = \frac{3}{1+2^{1-2m}}  \qquad \mbox{and hence} \qquad \lim_{m \to \infty} \Delta_{eff}(D_m) = \Delta_{eff}(D_\infty) = 3 \ .
\label{delta_dhl}
\eeq
This limit is approached exponentially rapidly as $m$ gets large.


\section{RG Transformations $F_q(v)$ and $r_q(y)$ and their fixed points}
\label{transformation}

\subsection{RG Transformations}

By carrying out the summation over the spins at intermediate vertices at each
stage, one finds the following iterative transformation for the partition
function of the Potts model on the Diamond Hierarchical Graphs $D_m$ 
\cite{gk,ddi}
\beq
Z(D_{m+1},q,v) = Z(D_m,q,v')(q+2v)^{2^{2m+1}}
\label{zdhl}
\eeq
where
\beq
v' = F_q(v) = \frac{v^2(2q+4v+v^2)}{(q+2v)^2} 
\label{vp}
\eeq
or equivalently, in terms of $y'=v'+1$, 
\beq
y' = r_q(y) = \bigg [ \frac{q+y^2-1}{q+2(y-1)} \bigg ]^2  \ . 
\label{yp}
\eeq
The two mappings are conjugate under the change of variables $y = v+1$.
The iterative transformation (\ref{vp}) (or equivalently, (\ref{yp})),
embodies the action of the real-space renormalization group action here.
Although we do not append a subscript to $v'$ or $y'$, it is
understood that these quantities are transformed at each iteration. We
denote $F_q^2(v)$ and $r_q^2(y)$ as the functional composition, i.e.,
$F_q^2(y) \equiv F_q(F_q(v))$ and $r_q^2(y) \equiv r_q(r_q(y))$,
respectively, and similarly for $F_q^m(v)$ and $r_q^m(y)$.  Note that
the transformation (\ref{vp}) is singular at $v=-q/2$ (and
equivalently, (\ref{yp}) is singular at $y=1-(q/2)$), which is a
physical antiferromagnetic value of $v$ if $q \in (0,2]$. 

For illustrative purposes, we record the expressions for $Z(D_m,q,v)$
for the first two values of $m$. Note that $D_1=C_4$, the circuit
graph with four vertices.  Elementary calculations yield
\beq
Z(D_0,q,v)=Z(T_2,q,v) = q(q+v)
\label{zdm0}
\eeq
and
\beqs
Z(D_1,q,v) &=& Z(C_4,q,v) = (q+v)^4 + (q-1)v^4  = q(q^3+4q^2v+6qv^2+4v^3+v^4) \ . 
\label{zdm1}
\eeqs
In passing, it may be remarked that many expressions are simpler when
written in terms of $v$ instead of $y$. This is the case for the basic
Eq.\ (\ref{cluster}), and, for example, $Z(D_1,q,v)$ consists of 5
terms when written as a polynomial in $q$ and $v$, but 10 terms when
written in terms of $q$ and $y$.  For this reason, we shall generally
express our results in terms of $q$ and $v$. However, some formulas
show an interesting structure when written in terms of $y$; for
example, $r_q(y)$ is a perfect square.  To keep matters simple, we will focus on
the $v$ variable throughout this paper and mention the $y$ variable only when it is helpful in an explanation.
In the special case $v=-1$, Eq.\ (\ref{zdm0}) and (\ref{zdm1}) yield
the chromatic polynomials
\beq
P(D_0,q) = q(q-1) \qquad \mbox{and} \qquad P(D_1,q)=q(q-1)(q^2-3q+3) \ .
\label{pm0}
\eeq
The chromatic polynomial $P(G,q)$ of any graph $G$ has a zero
at $q=0$, and the chromatic polynomial of any graph $G$ with at least
one edge has a zero at $q=1$.  As we shall show, for $D_m$ and
the limit $D_\infty$, the zero at $q=1$ is isolated, while the zero at
$q=0$ occurs at one of the points where the continuous locus ${\mathcal
  B}_q(-1)$ crosses the real $q$ axis.

The explicit expressions for $Z(D_m,q,v)$ become cumbersome to work with by
hand for $m \geq 2$, but computer algebra systems can work with them and solve
for the zeros of $Z(D_m,q,v_0)$ (fixed $v_0$) and of $Z(D_m,q_0,v)$ (fixed
$q_0$) for $m \leq 4$, at which point they become polynomials of degree $172$
and $256$, respectively.  (Plots of such zeros are shown on the right-hand
sides of Figures \ref{bq_vm1_figure}-\ref{dhlyplotq1000}).

\subsection{Preservation or reversal of the sign of $J$}

It will be convenient to write $F_q(v)$ as the product of one factor,
$F_{q,1}(v)$, that is positive-semidefinite for real values of $q$ and $v$,
times another, $F_{q,2}(v)$ that can have either sign:
\beq
F_q(v) = F_{q,1}(v) \, F_{q,2}(v) \ ,
\label{fproduct}
\eeq
where
\beq
F_{q,1}(v) = \bigg ( \frac{v}{q+2v} \bigg )^2 \qquad \mbox{and} \qquad F_{q,2}(v) = 2q+4v+v^2 \ .
\label{fq1}
\eeq
%
%
%

A basic question that one can ask about the RG transformation $F_q(v)$ is
whether it keeps the sign of $J$ invariant or reverses it. Recall that the
physical ranges of $v$ are $v > 0$ for the ferromagnetic sign, $J > 0$, and $-1
\le v < 0$ for the antiferromagnetic sign, $J < 0$, with $v=0$ corresponding to
$J=0$ or infinite temperature, $\beta=0$. Clearly, this transformation $F_q(v)$
maps $v=0$ to $v'=0$.  For nonzero $v$, as is evident from Eq.\ (\ref{vp}), the
sign of $v'$ is determined by the sign of the factor $F_{q,2}(v)$. Now, with
$q$ and $v$ real,
\beq
F_{q,2}(v) > 0 \ \Longleftrightarrow \ q > -\frac{v(4+v)}{2} \ .
\label{f2pos}
\eeq
For physical $q$ values, the inequality (\ref{f2pos}) holds for all
ferromagnetic couplings, i.e. for all $v > 0$.  However, the situation is
different for antiferromagnetic couplings, i.e., $v \in [-1,0]$.  As $v$
decreases from 0 to $-1$, the right-hand side of the inequality in
(\ref{f2pos}) increases from 0 to 3/2. So for
$q > 3/2$, which includes the usual integer values $q \ge 2$, $F_{q,2}(v)$ is
positive, so the RG transformation $F$ maps an antiferromagnetic $v$ to a
ferromagnetic $v' > 0$.  Further iterations of this RG transformation keep
the coupling ferromagnetic. It should also be noted that if one considers $q$
values down to, and including, $q=1$, then one must take account of the fact
that as $v$ decreases through AFM values to $-q/2$, $F_q(v)$ diverges.  This
divergence occurs in the physical AFM interval $v \in [-1,0]$ if
$0 < q \le 2$, which includes the  integer values 1 and 2.  In particular, for
the $q=2$ (Ising) case, $F_2(v)$ maps the limit $T \to 0$ for the
AFM (i.e., $v \searrow -1$) to the $T=0$ FM ($v'=\infty$).
For $q \ge 2$, a small real positive value of $v$ is mapped by $F_q(v)$ to a value
$v'$ that is smaller than $v$ (and positive).  Thus, for small positive $v$,
each iteration yields a smaller $v'$ and hence a higher temperature, so that
the fixed point is $v' \to 0^+$.

\subsection{RG Fixed Points}
\label{SUBSEC:RG_FP}

A particularly important set of values of $v$ is the set left invariant by the
transformation, i.e., the values that are RG fixed points.   Here we will focus on those fixed points
occurring for real values of $q$ and $v$.
They are obtained by solving
\begin{align}\label{EQN:FIXEDPTS}
F_q(v) = \frac{v^2(2q+4v+v^2)}{(q+2v)^2} = v
\end{align}
for the variable $v$ in terms of the parameter $q$.  Two fixed points, namely 
\begin{align*}
v=0 \qquad \mbox{and} \qquad v = \infty \ , 
\end{align*}
exist for every choice of $q$.  The fixed point at $v=0$ corresponds to
infinite-temperature (equivalently, zero-coupling), where the spin-spin
interaction has no effect. The fixed point $v = \infty$ can be interpreted as
meaning $\lim_{v \rightarrow \infty} F_q(v) = \infty$, and it corresponds to
zero temperature (equivalently infinite-coupling) in the ferromagnetic case.

Besides the fixed points at $v=0$ and $v=\infty$, solutions to the fixed
point equation (\ref{EQN:FIXEDPTS}) correspond to values of $v$ satisfying the
following cubic equation
\beq q^2+2qv-v^3=0 \ .
\label{dvpdvzero}
\eeq
The nature of the real roots of a cubic equation depend on the sign (or
vanishing) of its discriminant~$\Delta_3$ (see, e.g., \cite{disc}): (i) if
$\Delta_3 > 0$, then all of the roots of Eq. (\ref{dvpdvzero}) are real;
(ii) if $\Delta_3 < 0$, then Eq.  (\ref{dvpdvzero}) has one real root and a
complex-conjugate pair of roots; (iii) if $\Delta_3=0$, then at least two of
the roots of Eq. (\ref{dvpdvzero}) coincide. For Eq. (\ref{dvpdvzero}), 
the discriminant is 
 \beq \Delta_3 = -q^3(27q-32) \ .
\label{disc3}
\eeq
Let us consider the possibilities as $q$ increases from negative to positive
values:
\begin{itemize}
\item[(i)] For $q < 0$ we have $\Delta_3 < 0$ and (\ref{dvpdvzero}) has one real
  root and a complex-conjugate pair of non-real roots.  In particular, for
  these values of $q$, the RG transformation $F_q(v)$ has only one additional
  fixed point (other than $v=0,\infty$).

\item[(ii)] For $q=0$, Eq.\ (\ref{dvpdvzero}) has a triple root at $v=0$, which
  therefore coincides with the fixed point of $F_q(v)$ already discussed at the
  beginning of this subsection.  (Note that the degree of $F_q(v)$ drops from
  $4$ to $2$ at this parameter value, producing a dramatic change in the
  mapping.)

\item[(iii)] For $0 < q < \frac{32}{27}$, we have $\Delta_3 > 0$, and
  (\ref{dvpdvzero}) has three real solutions, corresponding to three additional
  real fixed points of $F_q(v)$.

\item[(iv)] For $q = \frac{32}{27}$, Eq.\ (\ref{dvpdvzero}) factorizes as 
 $-[v-(16/9)][v+(8/9)]^2=0$ and thus has $v=\{
  \frac{16}{9}, \ -\frac{8}{9}, \ -\frac{8}{9} \}$ as solutions. Hence, 
  $F_q(v)$ has these values as fixed points.

\item[(v)] For $q > \frac{32}{27}$ we again have $\Delta_3 < 0$, and so the same
  description as Case (i) applies.

\end{itemize}
We summarize this discussion with Figure \ref{FIG_FIXED_POINTS}.

\begin{figure}
\includegraphics[scale=0.4]{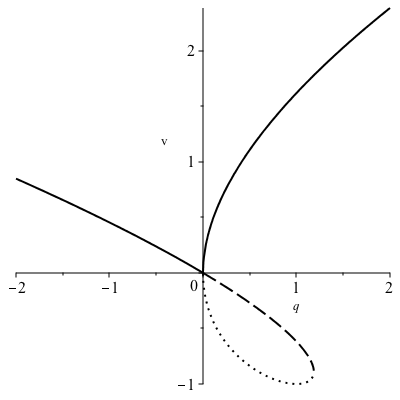}
\caption{Plot of the real fixed points of $F_q(v)$ other than $v=0$ and
  $v=\infty$.  The $v_{c,PM-FM}(q)$ fixed point exists for all $q \in
  \mathbb{R}$ and is denoted by a solid line.  The $v_+(q)$ fixed point exists
  for $0 \leq q \leq \frac{32}{27}$ and is depicted by a dashed line.  The
  $v_-(q)$ fixed point for $0 \leq q \leq \frac{32}{27}$ is depicted by a
  dotted line.
\label{FIG_FIXED_POINTS}}
\end{figure}

The additional fixed points of $F_q(v)$ have a special interpretation when 
$q = 32/27$. The fixed point at $v=16/9=1.77778$ is the critical
value of $v$ for a Potts ferromagnet with this value of $q=32/27$, and the
double root at $v=-8/9$ corresponds formally to a finite-temperature
antiferromagnet.  We use the word ``formally'' here, because for non-integral
$q$ the partition function of the Potts antiferromagnet does not, in general,
define a Gibbs measure and hence a normal statistical physics system.

Let us first discuss the case where $q > 32/27$ and focus on the real
root of Eq.\ (\ref{dvpdvzero}). This root occurs as a positive value
that we denote as $v_{c,PM-FM} = y_{c,PM-FM}-1$. If $q$ is an integer
$q \ge 2$, then, in the $m \to \infty$ limit,
there is a phase transition at $K_{c,PM-FM}=\ln(y_{c,PM-FM})$, i.e.,
at the temperature
\beq
T_{c,PM-FM} = \frac{J}{k_B \, \ln(y_{c,PM-FM})} \ ,
\label{tc_pmfm}
\eeq
from a paramagnetic (PM) phase with manifest $S_q$ symmetry at
high-temperatures $T > T_{c,PM-FM}$, i.e., $0 \le v \le v_{c,PM-FM}$,
to a low-temperature phase with ferromagnetic (FM) long-range order
(magnetization) and associated spontaneous symmetry breaking of the
$S_q$ symmetry to $S_{q-1}$ for $T < T_{c,PM-FM}$, i.e., $v >
v_{c,PM-FM}$.  We can generalize this to a phase transition from a
paramagnetic to ferromagnetic phase for $q$ not restricted to integers
$\ge 2$ but instead taking on any real value $q > 32/27$, since, as
discussed above, $Z(G,q,v)$ defines a Gibbs measure for real positive
$q$ and $v \ge 0$.

We recall that, on a regular lattice (in the thermodynamic limit), a standard
Peierls argument can be used to prove that a discrete spin model without
frustration, competing interactions, disorder, or dilution has a
finite-temperature phase transition if the spatial dimension of the lattice is
greater than the lower critical dimensionality, $d_\ell=1$.  Although a fractal
lattice is not a regular lattice in the conventional sense, arguments have been
given \cite{gk,gam} that if the Hausdorff dimensionality of the fractal lattice
(in the $n \to \infty$ limit) is greater than $d_\ell$, then a discrete model
(without frustration, competing interactions, disorder, or dilution) will also
have a finite-temperature phase transition on the fractal lattice.  This
conclusion applies to our present case, since the Hausdorff dimension
$d_{H}(D_\infty) = 2$ given in Eq.\ (\ref{dhldim}) is greater than 1. On a
regular lattice, the PM-FM phase transition in the $q$-state Potts model is
second order, with a divergent correlation length, if $q \in (0,4]$ and
first-order, with a nonzero latent heat, if $q > 4$ (e.g., \cite{wurev} and
references therein).  As discussed above, in this FM case, one can generalize
$q$ here from positive integers to positive real numbers for the ferromagnetic
case while retaining a Gibbs measure, and one may assume this generalization
here.  In the limit $q \to 0$, one must take account of a relevant
noncommutativity \cite{a}; one can define a nonvanishing free energy if one
takes $n \to \infty$ first, and then $q \to 0$, and with this order of limits,
there is again a second-order phase transition.  An alternate approach to the
$q=0$ case is to deal with the reduced partition function, $Z_r(D_m,q,v)$.

For $q \ge 32/27$, the physical root of Eq.\ (\ref{dvpdvzero}) is given by
(the real, positive number)
\beq
v_{c,PM-FM}(q) = y_{c,PM,FM}(q)-1 = 2^{-1/3}S + \frac{2^{4/3} \, q} {3 S} \ ,
\label{vc_pmfm}
\eeq
with
\beq
S \equiv ( q^2 + \sqrt{R_c} \ )^{1/3}  \qquad \mbox{and} \qquad R_c \equiv q^3\Big ( q- \frac{32}{27} \Big ) = -\frac{\Delta_3}{27} \ ,
\label{qcap}
\eeq
where the discriminant $\Delta_3$ was given in Eq.\ (\ref{disc3}).  The
corresponding physical temperature is given by Eq.\ (\ref{tc_pmfm}).  The
expression for $v_{c,PM-FM}$ in Eq.\ (\ref{vc_pmfm}) is a monotonically
increasing function of $q$ for $q \ge 32/27$. This is understandable
physically, since the larger $q$ is, the more statistical fluctuations there
are, so one must cool the system to a lower temperature, i.e., a larger value
of $K$ and hence $v$, for it to undergo the phase transition to a phase with
ferromagnetic order.  In addition to the special values $q=0$ and $q=32/27$ for
which Eq.\ (\ref{dvpdvzero}) factorizes into three linear factors, there are
also values of $q$ for which it factorizes into a linear factor times a
quadratic factor in $v$, so that the expressions for all of the roots simplify
considerably. For example, for $q=1$, Eq.\ (\ref{dvpdvzero}) factorizes as
$-(v+1)(v^2-v-1)=0$, with solutions $v=-1, \ (1 \pm \sqrt{5} \ )/2$; for $q=3$,
Eq.\ (\ref{dvpdvzero}) factorizes as $-(v-3)(v^2+3v+3)=0$, with solutions $v=3,
\ (-3 \pm i\sqrt{3} \ )/2$; and so forth for certain larger values of $q$.  We
list some illustrative values of $v_{c,PM-FM}(q)$ in Table
\ref{vc_pmfm_values}.
\begin{table}
\caption{\footnotesize{Illustrative values of $v_{c,PM-FM}$ as a function of
    $q$ for the $q$-state Potts model on the Diamond Hierarchical Lattice.}}
\begin{center}
\begin{tabular}{|c|c|}
\hline\hline
$q$    & $v_{c,PM-FM}$    \\ \hline
1      & $(1+\sqrt{5} \, )/2 = 1.6180$ \\
32/27  & 16/9 = 1.7777...  \\
2      & 2.3830   \\
3      & 3        \\
4      & 3.5386   \\
5      & 4.0261   \\
16     & 8        \\
\hline\hline
\end{tabular}
\end{center}
\label{vc_pmfm_values}
\end{table}

For our analysis, we also display the other two solutions in $v$ of the cubic
equation (\ref{dvpdvzero}),
\beqs
v_\pm(q) &=&-\frac{v_{c,PM-FM}}{2} \pm i \frac{\sqrt{3}}{2} \,
      \Big [ 2^{-1/3}S - \frac{2^{4/3} q}{3 S } \Big ] \cr\cr
        &=&-2^{-4/3}S -\frac{2^{1/3}q}{3S} \pm i \frac{\sqrt{3}}{2} \,
      \Big [ 2^{-1/3}S - \frac{2^{4/3} q}{3 S } \Big ]  \ , 
\label{v23}
\eeqs
where the subscripts $\pm$ correspond to the $\pm$ signs in
front of the factor of $i$.  Although we initially chose $v_{c,PM-FM}$ to be
the unique root of Eq.\ (\ref{dvpdvzero}) for $q \ge 32/27$, the formula
(\ref{vc_pmfm}) uniquely determines a fixed point of $F_q(v)$ for all $q \in
\mathbb{R}$, which we will continue to refer to as $v_{c,PM-FM}$.  We remark
that, extending to complex values, $|v_{c,PM-FM}| \sim |q|^{2/3}$ as $|q| \to
\infty$.

We refer the reader to Figure \ref{FIG_FIXED_POINTS} where the three fixed
points $v_-, v_+,$ and $v_{c,PM-FM}$ are labeled for varying $q$.  Remark that
they are ordered by
\begin{align*}
v_-(q) < v_+(q) < v_{c,PM-FM}(q)
\end{align*}
for those $q$ where all three exist and are distinct, i.e. for $0 < q <
\frac{32}{27}$.

The effect of the RG transformation Eq.\ (\ref{vp}) on $v$ can thus be
explained physically. For the high-temperature region of the
ferromagnet, $T > T_{c,PM-FM}$, i.e., the interval range $0 < v <
v_{c,PM-FM}$, the RG transformation Eq.\ (\ref{vp}) maps the initial value
of $v$ to a smaller value of $v'$, with the RG fixed point being
at $T=\infty$, i.e., $v=0$.  This is the standard
attractive infinite-temperature fixed point of the real-space
renormalization group in statistical physics and reflects the fact
that a critical point is a repulsive fixed point of the RG in the
temperature direction.  In turn, this is a consequence of the fact
that the real-space RG blocking transformation with a blocking factor
of $b$ reduces the correlation length as $\xi \to \xi/b$, which
is reduced finally to $\xi=0$ at $T=\infty$. If, in
contrast, the initial (real) value of the temperature is less than
$T_{c,PM-FM}$, i.e., $v$ is greater than
$v_{c,PM-FM}$, then the RG transformation Eq.\ (\ref{vp}) maps $v$ to a
larger value of $v'$, with the RG fixed point in this phase being the
zero-temperature fixed point where again the correlation length
(defined by the connected spin-spin correlation function) vanishes.

We remark that since $D_m$ is bipartite, the PM-AFM critical point is given
by $K_{c,PM-AFM}=-K_{c,PM-FM}$, i.e., $y_{c,PM-AFM} = (y_{c,PM-FM})^{-1}$.  
As $q$ increases above 32/27, the other two roots of Eq.\ (\ref{dvpdvzero}),
which formed a double root at $v=-8/9$, bifurcate into a complex-conjugate pair
of roots, which move to the upper and lower left away from the real axis.
Asymptotically,
\beq
v_{c,PM-FM}(q) \sim q^{2/3} \quad {\rm as} \quad q \to \infty \ .
\label{vcasymptotic}
\eeq

As noted above, for values of $q$ in the interval $0 \le q \le 32/27$,
Eq.\ (\ref{dvpdvzero}) has three real roots. As $q$ decreases below the
value 32/27, the real root $v_{c,PM-FM}$ decreases below 16/9, and the
double root at $v=-8/9$ bifurcates into two real roots, with $v_-$ 
decreasing and $v_+$ increasing as $q$ decreases toward $1$.
When $q$ reaches 1, Eq. (29) factorizes as $(v+1)(v^2-v-1)=0$ with
solutions $v=-1$ and $v_\pm = (1/2)(1 \pm \sqrt{5})$. Of these solutions, the
first, $v=-1$, corresponds formally to the $T=0$ Potts antiferromagnet; the
second, $v_-=-0.6180..$, corresponds to a finite-temperature Potts
antiferromagnet, and the third, $v_+ = v_{c,PM-FM} = 1.6180..$, corresponds to
a finite-temperature Potts ferromagnet.



\section{Rigorous Results from Complex Dynamics}
\label{complex_dynamics}

In order to prove rigorous results about the continuous accumulation loci
${\mathcal B}_q(v_0)$ for various values of the temperature-like Boltzmann
variable $v_0$ and to make computer pictures of it, we will need some results
from complex dynamics.  We will first describe them in a general context and
then specialize to the case of the renormalization mapping $F_q(v)$.

\subsection{Generalities on marked points and the passive/active dichotomy.}

Let $f_\lambda(z)$ denote a family of rational mappings $f_\lambda:
\widehat{\mathbb{C}} \rightarrow \widehat{\mathbb{C}}$ of the Riemann Sphere
$\widehat{\mathbb{C}}$ depending holomorphically on a complex parameter
$\lambda \in \Lambda$.  For our purposes, $\Lambda$ is an open subset of the
complex plane~$\mathbb{C}$.  Let $a(\lambda)$ be a choice of initial condition
for the iterates of $f_\lambda$ that depends holomorphically on the
parameter~$\lambda$.  It is called a ``marked point'', and historically
\cite{LY2,MSS} this theory was used for marked points that are critical points
of the rational mapping, in order to understand bifurcations of the mapping
itself. We remark that none of the results presented in this subsection are
new, and the proofs presented below are adaptations of those from the classical
papers to our current context and notations.

Following the terminology of McMullen \cite{muniv}, the marked point
$a(\lambda)$ is ``passive'' at parameter $\lambda_0$ if there is an open
neighborhood $U$ of $\lambda_0$ on which the the sequence of functions $\lambda
\to f_\lambda^m (a(\lambda))$ forms a normal family, in the sense of Montel's
Theorem (see, e.g., \cite{remmert} and references therein). As before, the
notation $f_\lambda(a(\lambda))$ denotes functional composition,
$f_\lambda^2(a(\lambda)) = f_\lambda(f_\lambda(a(\lambda)))$, and so forth for
higher~$m$.

The set of all passive parameters is open and is called the passive locus for
the marked point~$a(\lambda)$.  If the marked point $a(\lambda)$ is passive at
the parameter $\lambda=\lambda_0$, then the behavior of the initial condition
$a(\lambda)$ does not change much as $\lambda$ varies in a small neighborhood
of $\lambda_0$.  A parameter value $\lambda_0$ is defined as ``active'' if it
is not passive.  At these active parameter values, the initial condition
$a(\lambda)$ undergoes quite different dynamical behavior as $\lambda$ is
varied. Roughly speaking, this is analogous to the notion of bifurcation in the
theory of dynamical systems, but the global behavior of the mapping $f_\lambda$
need not change at $\lambda_0$, even if the marked point $a(\lambda)$ is active
at $\lambda=\lambda_0$.  One can also think of the passive locus as a
``parameter space analog of the Fatou set'' that is associated to the marked
point $a(\lambda)$ and the active locus as a ``parameter space analog of the
Julia set'' that is associated to the marked point $a(\lambda)$ 
\cite{juliaset}.

Let us describe a simple way for a parameter $\lambda_0$ to be active.  Suppose
that $z_\bullet$ is a repelling fixed point for $f_{\lambda_0}$.  Then
$z_\bullet$ can be holomorphically continued to be a repelling fixed point
$z_\bullet(\lambda)$ of $f_\lambda$ for all $\lambda$ in some neighborhood $U$
of $\lambda_0$.  One says that ``$f^{n_0}$ maps the marked point $a(\lambda)$
non-persistently onto the repelling fixed point  $z_\bullet(\lambda)$ at
parameter $\lambda_0$'' if
\begin{align*}
f^{n_0}_{\lambda_0}(a(\lambda_0)) = z_\bullet(\lambda_0) \qquad \mbox{and} \qquad 
f_{\lambda}(a(\lambda)) \not \equiv z_\bullet(\lambda) \quad \mbox{on} \quad  U.
\end{align*}
In this case, it is easy to show that $\lambda_0$ is an active parameter for
the marked point $a(\lambda)$ under~$f_\lambda$.  (The same holds if $f^{n_0}$
maps the marked point $a(\lambda)$ ``non-persistently'' onto a point from a
repelling periodic cycle at parameter $\lambda_0$.)

\begin{lemma}\label{LEM:NO_ISOLATED_ACTIVE_PARAM}
The set of active parameters of the marked point $a(\lambda)$ under the mapping $f_\lambda$ contains no isolated points.
\end{lemma}

\begin{proof}

  Suppose that $\lambda_0$ is an active parameter and $U$ is any neighborhood
  of $\lambda_0$.  Since repelling periodic points are dense in the Julia set
  $J(f_{\lambda_0})$ we can choose a repelling periodic cycle
  $z_1,z_2,\ldots,z_k$ of period $k \geq 3$ for $f_{\lambda_0}$ that is
  disjoint from $a(\lambda_0)$.  Restricting $U$ to a smaller neighborhood of
  $z_0$, if necessary, we can suppose that this periodic cycle varies
  holomorphically as
\begin{align}\label{HOL_VARYING_REP_CYCLE}
z_1(\lambda),z_2(\lambda),\ldots,z_k(\lambda),
\end{align}
forming a repelling cycle of period $k$ of $f_\lambda$ for all $\lambda \in U$.
Since $\lambda_0$ is active, Montel's Theorem implies that there is some
parameter $\lambda_1 \in U$ and some iterate $n_0$ such that
$f_{\lambda_1}^{n_0}(a(\lambda_1)) = z_j(\lambda_1)$ for some $1 \leq j \leq
k$.  Since $a(\lambda_0) \neq z_\ell(\lambda_0)$ for all $1 \leq \ell \leq k$
we see that $f^{n_0}$ maps the marked point $a(\lambda)$ non-persistently onto
the repelling periodic cycle from Eq.\ (\ref{HOL_VARYING_REP_CYCLE}) at
parameter $\lambda_1$.  Therefore, $\lambda_1 \in U$ is another active
parameter.
\end{proof}

Recall that a point $b \in \widehat{\mathbb{C}}$ is ``exceptional'' for a
rational map $f: \widehat{\mathbb{C}} \rightarrow \widehat{\mathbb{C}}$ if the
cardinality of the set $\{z \in \mathbb{C} \, : \, f^m(z) = b \mbox{ for some
  $m \geq 0$}\}$ is one or two.  A marked point $b(\lambda)$ is ``persistently
exceptional'' for a holomorphic family of rational maps $f_\lambda$ if for
every parameter $\lambda$ the point $b(\lambda)$ is exceptional
for~$f_\lambda$.

The key statement we need from holomorphic dynamics is the following simple
lemma:

\begin{lemma}\label{KEY_LEMMA}
Suppose $f_\lambda$ is a holomorphically varying family of rational maps, that $a(\lambda)$ and $b(\lambda)$ are
marked points, and that $b(\lambda)$ is not persistently exceptional for $f_\lambda$.

Then, if $\lambda_0$ is an active parameter for $a(\lambda)$ under $f_\lambda$ we have
\begin{align*}
\lambda_0 \in \overline{\{\lambda \in \Lambda \, : \, f^m(a(\lambda)) = b(\lambda) \, \mbox{ for some $m \geq 0$.}\}}.
\end{align*}
\end{lemma}

\noindent
The proof is classical, but we include it here for the convenience of the reader, closely following \cite[Proposition 3.5]{LYUBICH}.

\begin{proof}
Suppose for contradiction that $a(\lambda)$ is active under $f_\lambda$ at parameter $\lambda_0$
and that there is some neighborhood $U$ of $\lambda_0$ such that 
\begin{align}\label{ASSUMPTION_FOR_CONTRA}
f_\lambda^m(a(\lambda)) \neq b(\lambda) \quad \mbox{for all $\lambda \in U$ and all $n \geq 0$.}
\end{align}
Since $b(\lambda)$ is not persistently exceptional, there is some $\ell \geq 1$
such that $f_\lambda^{-\ell}(b(\lambda))$ contains $k \geq 3$ points for 
all but finitely many possible values of $\lambda$, for which it has fewer than $k$ points.
If $\lambda_0$ is not one of those values, we can work in sufficiently small
neighborhood $U' \subset U$ of $\lambda_0$ over which $f_\lambda^{-\ell}(b(\lambda))$ 
consists of $\ell$ disjoint graphs of holomorphic functions of $\lambda$.  Together
with Assumption~(\ref{ASSUMPTION_FOR_CONTRA}) and Montel's Theorem, this implies
that $\lambda \mapsto f_\lambda^m(a(\lambda))$ forms a normal family on $U'$, contrary to
the hypothesis that $\lambda_0 \in U'$ is active.

If it happened that $\lambda_0$ is one of the finitely many parameters for which  $f_{\lambda}^{-\ell}(b(\lambda))$ has fewer than $k$ points,
then we use the fact that active parameters are not isolated (Lemma \ref{LEM:NO_ISOLATED_ACTIVE_PARAM}) 
to replace $\lambda_0$ with another active parameter $\lambda_1$ for which $f_{\lambda_1}^{-\ell}(b(\lambda_1))$ has the maximal
number of preimages $k$.  We then apply the reasoning from the previous paragraph to the new active parameter $\lambda_1$.
\end{proof}

The following general classification of the types of behavior of a marked point
$a(\lambda)$ for the case in which $\lambda$ is in the passive locus will be
helpful for our discussion:
\begin{CPL}
  Let $f: \Lambda \times \mathbb{P}^1 \rightarrow \mathbb{P}^1$ be a
  holomorphic family and let $a(\lambda)$ be a marked point. Assume $U \subset
  \Lambda$ is a connected open subset where $a(\lambda)$ is passive. Then
  exactly one of the following cases holds:

\begin{enumerate}

\item[\rm (i)] $a(\lambda)$ is never preperiodic in $U$. In this case the
  closure of the orbit of $a(\lambda)$ can be followed by a holomorphic motion.

\item[\rm (ii)] $a(\lambda)$ is persistently preperiodic in $U$.

\item[\rm (iii)] There exists a persistently attracting (possibly
  superattracting) cycle attracting $a(\lambda)$ throughout $U$ and there is a
  closed subvariety $U' \subsetneq U$ such that the set of parameters $\lambda
  \in U \setminus U'$ for which $a(\lambda)$ is preperiodic is a proper closed
  subvariety in $U \setminus U'$.

\item[\rm (iv)] There exists a persistently irrationally neutral periodic point
  such that $a(\lambda)$ lies in the interior of its linearization domain
  throughout $U$ and the set of parameters $\lambda \in U$ for which
  $a(\lambda)$ is preperiodic is a proper closed subvariety in $U$.

        \end{enumerate}
\end{CPL}
\noindent
We do not include a proof of this rather difficult theorem.

\subsection{Application to $F_q(v)$ and the $q$-plane zeros ${\mathcal B}_q(v_0)$.}
\label{SUBSEC_USING_ACTIVE_PARAMS}
We will now explain how to use the techniques from the previous subsection to study the $q$-plane zeros for the DHL.
(The reader may wish to compare this discussion with that from 
\cite{rr}, but note that in that paper the variable $y=v+1$ is used
instead of the equivalent variable $v$.)

Recall that $q=0$ is always a zero for $Z(D_m,q,v)$.  The renormalization
procedure from Eqs.\ (\ref{zdhl}) and (\ref{vp}) implies that for any $v_0 \in
\mathbb{C}$ and any $q \neq 0$ we have
\begin{align*}
Z(D_m,q,v_0) = 0 \qquad {\rm if \ and \ only \ if} \qquad F_q^m(v_0) = -q,
\end{align*}
where
\begin{equation*}
F_q(v) = \frac{v^2(2q+4v+v^2)}{(q+2v)^2}
\end{equation*}
is the renormalization mapping (\ref{vp}).  Therefore, for our purposes, the
marked point $a(q)$ will correspond to the desired choice of $q$-plane and will
be constant: $a(q) \equiv v_0$.  (For example $a(q)~\equiv~-1$ will correspond
to the case of chromatic zeros.) Meanwhile, the other marked point will be
$b(q)~=~-q$.

We remark that the degree of $F_q(y)$ drops from $4$ to $2$ when $q=0$ due to
the appearance of a common factor of $v$ in the numerator and denominator.
This is the only parameter where such a drop in degree occurs, and therefore,
$F_q(v)$ is a holomorphic family of rational maps with parameter space $\Lambda
= \mathbb{C} \setminus \{0\}$.  In fact, the entire discussion in the remainder
of this section will only pertain to $q \in \mathbb{C} \setminus \{0\}$.

For any $v_0 \in \mathbb{C}$ let us denote the active locus of the marked point $a(q) \equiv v_0$ by ${\mathcal A}_q(v_0)$.

\begin{lemma}\label{LEM:ACTIVE_SUBSET_BQ}
Within $\mathbb{C} \setminus \{0\}$ we have ${\mathcal A}_q(v_0) \subset {\mathcal B}_q(v_0)$.
\end{lemma}

\begin{proof}
When $q=2$ we have $F^{-1}_q(b(q)) = F^{-1}_2(-2)=\{-1 \pm \sqrt{2}i \pm 1 i\}$ (four different values), so
that $b(q) = -q$ is not persistently exceptional for $F_q(v)$.  Therefore,
Lemma \ref{KEY_LEMMA} implies that any point of ${\mathcal A}_q(v_0)$ is in the
accumulation set of solutions to $F_q^m(v_0) = -q$ and hence the accumulation locus of zeros for
$Z(D_m,q,v_0)$, both considered as $m \rightarrow \infty$.  By Lemma
\ref{LEM:NO_ISOLATED_ACTIVE_PARAM}, ${\mathcal A}_q(v_0)$ contains no isolated
points, so we conclude that ${\mathcal A}_q(v_0) \subset {\mathcal B}_q(v_0)$.
\end{proof}

We will now make some observations that will help when drawing computer
pictures of ${\mathcal A}_q(v_0)$ and that will also help us to further relate
${\mathcal A}_q(v_0)$ to ${\mathcal B}_q(v_0)$.

For every complex parameter $q \neq 0$ we have that:
\begin{itemize}
\item[(i)] The point $v=0$ is a superattracting fixed point for
$F_q$  (i.e., $F_q(0) = 0$ and $F_q'(0) = 0$).
\item[(ii)] The point $v=\infty$ is a superattracting fixed point for
$F_q$  (i.e., $F_q(\infty) = \infty$ and $F_q'(\infty) = 0$, with the derivative computed in suitable local coordinates centered at $\infty$.)
\end{itemize}
Therefore, for every $q \neq 0$ there is an open neighborhood
$\mathcal{W}_q^s(0)$ consisting of initial conditions $v_0$ whose orbit under
$F_q^m$ converges to $0$ and similarly an open neighborhood
$\mathcal{W}_q^s(\infty)$ consisting of initial conditions $v_0$ whose orbit
under $F_q^m$ converges to $\infty$.  These neighborhoods depend continuously
on the parameter $q$.

Let
\begin{align*}
\mathcal{P}_q^0(v_0) := \{q \in \mathbb{C} \setminus \{0\} \, : \, F_{q}^m(v_0) \rightarrow 0\} \qquad \mbox{and} \qquad  \mathcal{P}_q^\infty(v_0) := \{q \in \mathbb{C} \setminus 0 \, : \, F_{q}^m(v_0) \rightarrow \infty\}.
\end{align*}
Each is a subset of the passive parameters for the marked point $a(q) \equiv v_0$ and each is an open set.

\begin{lemma}\label{LEM:BASIN_BOUNDARY}
Within $\mathbb{C} \setminus \{0\}$ we have
\begin{align*}
{\mathcal A}_q(v_0) = \partial \mathcal{P}_q^0(v_0) = \partial \mathcal{P}_q^\infty(v_0).
\end{align*}
(Here, $\partial$ denotes the topological boundary.)
\end{lemma}

\begin{proof}
The proof is similar to that of Lemma \ref{LEM:NO_ISOLATED_ACTIVE_PARAM}.  Suppose for contradiction that $q_0 \in {\mathcal A}_q(v_0)$ and suppose there
were an $\epsilon > 0$ such that the disc $\mathbb{D}_\epsilon(q_0)$ were disjoint from $\mathcal{P}_q^0(v_0)$.
Then, there exist values $v = v_\bullet$ and $v = v_*$ close to $v = 0$ such that $v_\bullet, v_* \in {\mathcal W}^s_q(0)$ for all $q \in \mathbb{D}_\epsilon(q_0)$.
In this case the family of holomorphic functions
\begin{align}\label{EQN_SEQ}
q \mapsto F_q^m(v_0)
\end{align}
defined on $\mathbb{D}_\epsilon(q_0)$ would omit the three values $v=0,
v=v_\bullet$ and $v=v_*$.  Montel's theorem would then imply that this is a normal
family on $\mathbb{D}_\epsilon(q_0)$, contradicting that $q_0$ is active.
Therefore ${\mathcal A}_q(v_0) \subset \partial \mathcal{P}_q^0(v_0)$.  

Now suppose that $q_0 \in \partial \mathcal{P}_q^0(v_0)$.  If $q_0$ were a passive parameter, then there is an $\epsilon > 0$ such
that the sequence of functions given in Eq.\ (\ref{EQN_SEQ}) forms a normal family on $\mathbb{D}_\epsilon(q_0)$.
However, $\mathbb{D}_\epsilon(q_0)$ contains points of $\mathcal{P}_q^0(v_0)$, so the identity theorem for holomorphic functions
implies that the sequence of functions given in Eq.\ (\ref{EQN_SEQ}) converges to $0$ on all of $\mathbb{D}_\epsilon(q_0)$.
Therefore, $q_0$ is in the interior of $\mathcal{P}_q^0(v_0)$ contradicting the hypothesis that $q_0 \in  \partial \mathcal{P}_q^0(v_0)$.  We conclude that $q_0 \in {\mathcal A}_q(v_0)$ and hence
that $\partial \mathcal{P}_q^0(v_0) \subset {\mathcal A}_q(v_0)$. 

Proof that ${\mathcal A}_q(v_0) = \partial \mathcal{P}_q^\infty(v_0)$ is the same, so we omit it.
\end{proof}

\begin{lemma}\label{LEM:BASINS_NON_EMPTY}
For any $v_0 \neq 0,-4$ we have that 
\begin{align*}
\mathcal{P}_q^0(v_0) \neq \emptyset \qquad \mbox{and} \qquad \mathcal{P}_q^\infty(v_0) \neq \emptyset.
\end{align*}
In particular $\mathcal{A}_q(v_0) \neq \emptyset$ and the marked point $a(q) \equiv v_0$ is not persistently preperiodic. (I.e.\ Case~(ii) of the Dujardin-Favre criterion cannot hold on any connected component of the passive locus.)
\end{lemma}

\begin{proof}
  If $v_0 \neq 0,-4$ then the point $q_0 = -\frac{1}{2} v_0^2-2v_0 \in
  \mathbb{C} \setminus \{0\}$ and it satisfies $F_{q_0}(v_0) = 0$, implying
  that $q_0 \in \mathcal{P}_q^0(v_0)$.  If $v_0 \neq 0$ then
  $q_\infty = -2 v_0 \in \mathbb{C} \setminus \{0\}$ and it satisfies
  $F_{q_\infty}(v_0) = \infty$, implying that $q_\infty \in
  \mathcal{P}_q^\infty(v_0)$. Since both $\mathcal{P}_q^0(v_0)$ and
  $\mathcal{P}_q^\infty(v_0)$ are non-empty open subsets of the connected space
  $ \mathbb{C} \setminus \{0\}$, their boundaries are non-empty, implying that
  $\mathcal{A}_q(v_0) \neq \emptyset$.

If the marked point $a(q) \equiv v_0$ were to be preperiodic on some connected component of the passive locus, then 
it would be preperiodic to the same periodic cycle for all $q \in \mathbb{C} \setminus \{0\}$, by the identity theorem for holomorphic functions.  However,
there are parameters $q_0$ and $q_\infty$ for which $a(q) \equiv v_0$ maps onto the two different fixed points $0$ and $\infty$, 
so this is impossible.
\end{proof}

The conditions that $q \in  \mathcal{P}_q^0(v_0)$ or that $q \in
\mathcal{P}_q^\infty(v_0)$ are easy to check numerically on the computer, so
Lemma \ref{LEM:BASIN_BOUNDARY} makes visualization of the active locus
${\mathcal A}_q(v_0)$ easy.  See, for example, Figures
\ref{bq_vm1_figure}-\ref{bq_v99_figure}.

\begin{lemma}\label{LEM:BQ_NOT_SUBSET_P}
Within $\mathbb{C} \setminus \{0\}$ we have that
\begin{align*}
{\mathcal B}_q(v_0) \subset \mathbb{C} \setminus (\mathcal{P}_q^0(v_0) \cup \mathcal{P}_q^\infty(v_0)).
\end{align*}
\end{lemma}

\begin{proof}
Over any compact $K \subset \mathcal{P}_q^0(v_0) \setminus \{0\}$ there is a uniform
$M$ such that for all $m \geq M$ we have that $F_q^m(v_0)$ is sufficiently close to $0$ that it is not equal to $b(q) = -q$.  (Here we used that $q \neq 0$.)
This implies that ${\mathcal B}_q(v_0)$ is disjoint from $\mathcal{P}_q^0(v_0)$.  The same
holds for $\mathcal{P}_q^\infty(v_0)$ using identical reasoning.
\end{proof}

The following proposition summarizes Lemmas \ref{LEM:ACTIVE_SUBSET_BQ} - \ref{LEM:BQ_NOT_SUBSET_P}:

\begin{proposition}\label{PROP_DESCRIBING_FIGURES}
Within $\mathbb{C} \setminus \{0\}$ 
the continuous accumulation locus ${\mathcal B}_q(v_0)$ of zeros for
$Z(D_m,q,v_0)$ as $m\rightarrow \infty$ contains all parameters in the boundary
of $\mathcal{P}_q^0(v_0)$ and all parameters in the boundary of $\mathcal{P}_q^\infty(v_0)$ (these
boundaries are equal).  Moreover, ${\mathcal B}_q(v_0)$ is disjoint from $\mathcal{P}_q^0(v_0) \ \cup \mathcal{P}_q^\infty(v_0)$.
\end{proposition}

\subsection{On the possibility that ${\mathcal A}_q(v_0)
\subsetneq {\mathcal B}_q(v_0)$}

For a given choice of $v_0$ it may be possible to have ${\mathcal A}_q(v_0)
\subsetneq {\mathcal B}_q(v_0)$ as a result of there being components of the
passive locus for $a(q) \equiv v_0$ other than those in $\mathcal{P}_q^0(v_0)
\, \cup \, \mathcal{P}_q^\infty(v_0)$.  In our computer-generated pictures
(Figures \ref{bq_vm1_figure}-\ref{bq_v99_figure}) for many values of $v_0$ one
can see subsets of points colored black, corresponding to the condition that
$F_q^m(v_0) \not \rightarrow 0$ and $F_q^m(v_0) \not \rightarrow \infty$.  Any
connected component of the interior of this black subset will be a passive
component for the marked point $a(q)$.  We will see that some of these
components correspond to the marked point having orbit attracted to attracting
periodic orbits for $F_q$ other than $v=0$ and $v=\infty$.  For such
components, one can use similar reasoning to the proof of Proposition
\ref{PROP_DESCRIBING_FIGURES} to rule out points of ${\mathcal B}_q(v_0)$.
However, the other behaviors (i) and (iv) described in the Dujardin-Favre
classification of the passive locus could lead to points of ${\mathcal
  B}_q(v_0)$ that are within these black regions.  (Behavior (ii) is ruled out
by Lemma~\ref{LEM:BASINS_NON_EMPTY}.)

This is similar to the situation for the continuous accumulation locus of zeros
${\mathcal B}_v(q_0)$ in the $v$-plane for a fixed $q_0$.  So long as $v=-q_0$
is not an exceptional point of $F_{q_0}$ it follows from Montel's theorem the
Julia set of $F_{q_0}$ satisfies $J(q_0) \subset {\mathcal B}_v(q_0)$.
However, it can be possible to have $J(q_0) \subsetneq {\mathcal B}_v(q_0)$.
For example, this will happen if:
\begin{itemize}
\item[(i)]  $F_{q_0}$ has a Siegel Disc $D$ in its Fatou set (i.e.\ a component of the Fatou set on which $F_q$ is conjugate
to an irrational rotation),
\item[(ii)]  $v=-q_0 \in D$, and 
\item[(iii)] $v=-q_0$ is not equal to the unique fixed point of $F_q$ that is in $D$.
\end{itemize}
Then the closure of its forward orbit of $v$ will form a simple closed curve
$\gamma \subset D$, because the dynamics of $F_q$ is conjugate to an irrational
rotation on~$D$.  This leads to some of the iterated preimages of $v$ under
$F_{q_0}$ accumulating on every point of $\gamma$, implying that $\gamma
\subset {\mathcal B}_v(q_0) \setminus J(q_0)$.

For the $v$-plane zeros, this issue can be handled by considering only the
locus ${\mathcal B}'_v(q_0)$ where a positive proportion of the zeros of
$Z(D_m,q_0,v)$ accumulate.  So long as $-q_0$ is not an exceptional point for
$F_{q_0}(v)$ it is a consequence of the Lyubich and Friere-Lopes-Ma\~{n}e
Theorems \cite{Ly1,Ly2,FLM} that ${\mathcal B}'_v(q_0) = J(q_0)$.  In other words, any
possible zeros of $Z(D_m,q_0,v)$ occurring in the Fatou set of $F_{q_0}$ do so
with arbitrarily small proportion, in the limit $m \rightarrow \infty$.

This quantitative approach can also be taken in the $q$-planes.  Consider the
locus ${\mathcal B}'_q(v_0) \subset {\mathcal B}_q(v_0)$  where a positive
proportion of the zeros from $Z(D_m,q,v_0)$ accumulate.  For rational $v_0$
(and even any algebraic number $v_0$) it is a consequence of Theorem C' from
\cite{rr} that ${\mathcal A}_q(v_0) = {\mathcal B}'_q(v_0)$.  In other words,
any possible zeros of $Z(D_m,q,v_0)$ occurring in the passive locus of $a(q)
\equiv v_0$ for $F_q$ do so with arbitrarily small proportion, in the limit $m
\rightarrow \infty$.

\subsection{On thinking in $\mathbb{C}^2$}
\label{SUBSEC:C2}

It is very helpful to think about the continuous accumulation loci of zeros for
$Z(D_m,q,v)$, $m \rightarrow \infty$, as being a single object in
$\mathbb{C}^2\
$, with the loci
$B_q(v_0)$ being horizontal slices and $B_v(q_0)$ being
vertical slices of the same object.  This allows one to gain insight about
$B_q\
(v_0)$ by looking
at $B_{v_0}(q)$ near $q=q_0$ and vice-versa.

While this is very good intuition, additional care must be taken to ensure the
\
results are rigorous, as the following delicate example shows.
It is very natural to expect that
\begin{align}\label{BELIEF}
q_0 \in {\mathcal B}_q(v_0) \qquad \mbox{if and only if} \qquad v_0 \in
{\mathcal B}_v(q_0).
\end{align}
However, there are some delicate situations where this potentially might not
hold. Suppose that there is a component ${\mathcal P}$ of the passive locus for
the marked point $a(q) \equiv v_0$ under $F_q(v)$ for which $a(q) \in J(q)$
(Julia set for $F_q$) for all $q \in {\mathcal P}$.  (This would correspond to
Case (i) of the Dujardin-Favre classification, since we've ruled out Case (ii)
by Lemma \ref{LEM:BASINS_NON_EMPTY}.) Then, $v_0 \in {\mathcal B}_v(q)$ for all
$q \in {\mathcal P}$ because $J(q) \subset {\mathcal B}_v(q)$.

If there are some parameters $q \in {\mathcal P}$ for which $b(q) = -q$ is in
the basin of attraction of an attracting cycle of $F_q$, then such parameters
would form an open subset ${\mathcal P}_0$ of ${\mathcal P}$.  For any $q \in
{\mathcal P}_0$ we cannot have $F_q^m(a(q)) = b(q)$ because the Julia set and
Fatou set are complementary and invariant under $F_q$.  We would therefore have
$Z(D_m,q,v_0) \neq 0$ on the open set ${\mathcal P}_0$ for all $m \geq 0$.  In
particular, such $q$ are not in ${\mathcal B}_q(v_0)$ and Eq.\ (\ref{BELIEF})
would fail.

We do not have an explicit example of this problematic behavior for the RG
mapping $F_q(v)$.  However, to avoid such delicate issues, we will work with
the techniques from Section \ref{SUBSEC_USING_ACTIVE_PARAMS} involving the
active parameters ${\mathcal A}_q(v_0)$.  (See also Appendix \ref{q0_appendix}
for discussion of an additional problem that happens when $q=0$.)


\section{Zeros in the $q$ Plane for $v=-1$: Chromatic Zeros}
\label{chromaticzeros}

In this section we study the chromatic polynomial $P(D_m,q)=Z(D_m.q,-1)$ and
its zeros in the complex $q$ plane, i.e., the chromatic zeros of $D_m$ and
their continuous accumulation set ${\mathcal B}_q(-1)$ for $m \to \infty$. The
left side of Figure \ref{bq_vm1_figure} shows a computer-generated image of the
regions
\begin{align*}
{\mathcal P}_q^0(-1) &:= \{q \in \mathbb{C} \setminus \{0\} \, : \, F_q^m(-1) \rightarrow 0\} \qquad \mbox{colored white, and} \\
{\mathcal P}_q^\infty(-1) &:= \{q \in \mathbb{C} \setminus \{0\} \, : \, F_q^m(-1) \rightarrow \infty \} \qquad \mbox{colored blue.}
\end{align*}
Any point that is not in one of these two sets is colored black.  Because the
sets ${\mathcal P}_q^0(-1)$ and ${\mathcal P}_q^\infty(-1)$ are open, the set
of points colored black is closed.  According to Proposition
\ref{PROP_DESCRIBING_FIGURES}, ${\mathcal B}_q(-1)$ contains any point of the
boundary between the white, blue, and black sets.  It may contain additional
points, but they are necessarily in the interior of the black set.
The right side of Figure \ref{bq_vm1_figure} shows a plot of the 171 zeros of
$Z_r(D_m.q,-1)$ computed numerically in Mathematica.  (We have omitted the zero
at $q=0$, hence the subscript $r$ indicating that we consider the reduced
partition function.)
%
\begin{figure}
  \begin{center}
\scalebox{1.2}{
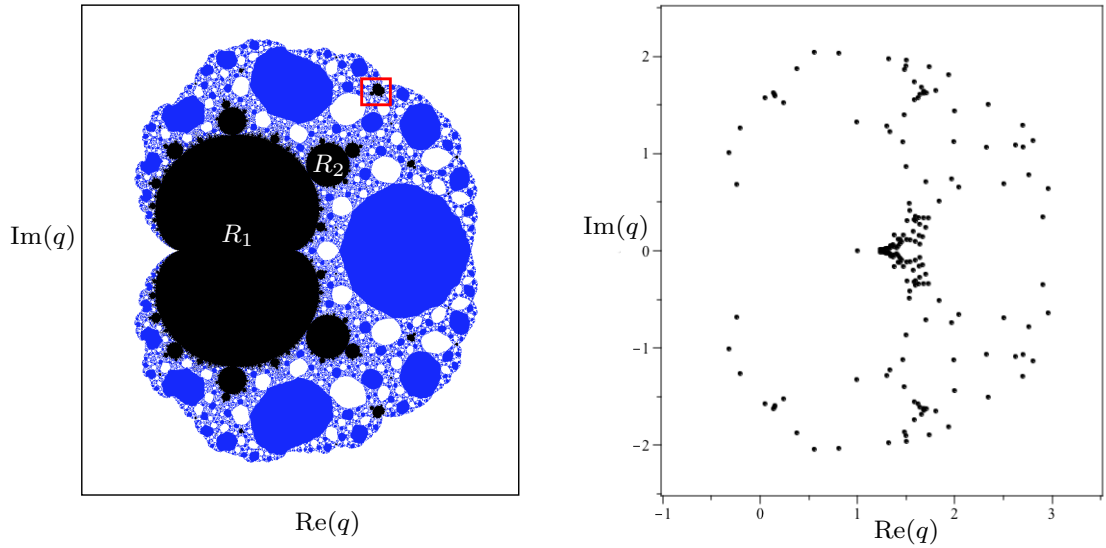
}
    \end{center}
    \caption{Left: Region diagram for $D_\infty$ in the complex $q$ plane for
      $v=-1$.  The locus ${\mathcal B}_q(-1)$ contains the boundaries between
      the white, blue, and black regions.  It may also contain additional
      points in the black regions.  (See Proposition
      \ref{PROP_DESCRIBING_FIGURES} and the paragraphs following it for further
      discussion.)  The red box is referred to in the caption to Figure
      \ref{bq_vm1_zoom_figure}.  Right: Zeros of the reduced chromatic
      polynomial $P_r(D_4,q)=Z_r(D_4,q,-1)$ (171 zeros).  Both left and right
      figures depict $-1 < {\rm Re}(q) < 3.5$ and $-2.5 < {\rm Im}(q) < 2.5$. }
  \label{bq_vm1_figure}
\end{figure}
%


\begin{figure}
  \begin{center}
\scalebox{1.2}{
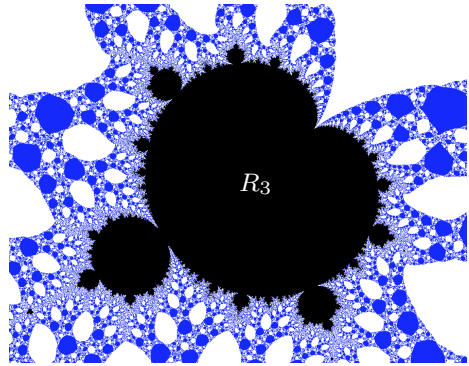
}
    \end{center}
    \caption{Magnified view of area surrounded by the small red box from the
      region diagram on Figure \ref{bq_vm1_figure}.  Here, $v=-1$ and $1.9 \leq
      {\rm Re}(q) \leq 2.15$ and $1.55 \leq {\rm Im}(q) \leq 1.75$.}
  \label{bq_vm1_zoom_figure}
\end{figure}
%


According to Lemma \ref{LEM:BASIN_BOUNDARY} the interior of the black set is
contained in the passive locus for the marked point $a(q) \equiv -1$ under
$F_q$.  Let us consider what happens on the connected components $R_1, R_2,$
and $R_3$ of the interior of the black set that are labeled on Figures
\ref{bq_vm1_figure} and \ref{bq_vm1_zoom_figure}.

We begin with the component labeled $R_1$, containing the point $q=0.5$.  It
intersects the real $q$ axis in the interval $(0,32/27)$.  For $q \in
(0,32/27)$, the fixed point $v_{-}(q)$ given in Eq. (\ref{v23}) is an
attracting fixed point for $F_q$ that differs from $0$ and $\infty$.
Furthermore, this fixed point $v_{-}(q)$ can be analytically continued for all
$q \in R_1$ and it remains an attracting fixed point of $F_q$ at these $q$
values.  For all $q \in R_1$ the orbit of initial condition $a(q) \equiv -1$
converges to $v_{-}(q)$, which is an example of passive behavior for this
marked point.

Now consider component $R_2$.  We have performed numerical computer studies
that indicate that for every $q \in R_2$, the RG mapping $F_q$ has a periodic
orbit of period 3 that is attracting.  For values of $q \in R_2$, the orbit of
$a(q) \equiv -1$ under $F_q$ converges to this attracting cycle.

Finally, consider component $R_3$, which is the cardioid of the ``baby
Mandelbrot set'' in the upper right corner of the left half of
Fig. \ref{bq_vm1_figure} that is also shown in a magnified view in
Fig. \ref{bq_vm1_zoom_figure}.  Numerical experiments show that for all $q \in
R_3$, the RG mapping $F_q$ has an attracting periodic cycle of period~2, and
that the marked point $a(q) \equiv -1$ has orbit converging to it.

The complexity of this region diagram is evident, even with the finite
resolution of Fig. \ref{bq_vm1_figure}.  One reason for the complexity is the
appearance of baby Mandelbrot sets.  We show one of them in
Fig.~\ref{bq_vm1_zoom_figure}, corresponding to a magnified view of the
region enclosed by a red box in Fig. \ref{bq_vm1_figure}.  However,
there are infinitely many baby Mandelbrot sets in Fig. \ref{bq_vm1_figure}.
Their existence can be explained using complex dynamics renormalization theory:

\medskip
\begin{theorem}[McMullen 1997 \cite{muniv}]
  Suppose $f_\lambda(z)$ is a holomorphic family of rational maps and
  $c(\lambda)$ is a marked critical point for $f_\lambda(z)$.  If there is at
  least one parameter $\lambda_0$ so that the marked critical point
  $c(\lambda)$ is active under $f_\lambda(z)$, then the active locus for the
  marked critical point $c(\lambda)$ contains quasiconformal copies of the the
  Mandelbrot set (or possibly of the degree $d > 2$ generalization thereof).
\end{theorem}
\noindent
In particular, it follows from the work of Shishikura \cite{Shi1} that the active locus of $c(\lambda)$ has Haudorff dimension equal to 2.

With regard to the application of McMullen's Theorem to our $F_q(v)$ in Eq.\
(\ref{vp}) with marked point $a(q) \equiv -1$, one must note that $a(q) = -1$
is not a critical point for $F_q$. However, $c(q) = -1 + \sqrt{1-q}$ is a
marked critical point for $F_q$, and one can check that
\beq
F_q(c(q)) = F_q\left(-1 + \sqrt{1-q}\right) \equiv -1 \equiv a(q) \ ,
\label{rd}
\eeq
so that $a(q)$ and $c(q)$ have the same orbit under $F_q^m$, which, in turn,
implies that their active and passive loci are the same.  Therefore, McMullen's
theorem implies that the limiting locus of chromatic zeros for the DHL has
Hausdorff dimension equal to 2 and that it contains small copies of the
Mandelbrot set (as seen in the figures).  (This was first observed by Chio and
Roeder \cite[Theorem B]{rr}.)

An important remark here is that it is essential that the marked point be a
critical point for~$F_q$.  So, the McMullen theorem does not apply to any of
the other slices with constant $v_0 \ne -1$ that we are considering.  In
particular, it does not imply that the limiting locus of zeros, ${\mathcal
  B}_q(v_0)$, for these other slices have Hausdorff dimension 2. This also
explains the absence of small Mandelbrot sets in Figures \ref{bq_vm0p8_figure}
- \ref{bq_v99_figure} for $v \ne -1$.

In fact, one can prove that the activity locus ${\mathcal A}_q(-1)$ for the
marked point $a(q) \equiv -1$ under $F_q(v)$ coincides with the bifurcation
locus ${\mathcal M}$ of the mapping.  This is done in Proposition 7.2 from
\cite{rr} for the mapping $r_q(y)$ and marked point $y = a(q) = 0$, which
correspond under conjugacy to the situation here.    Therefore, the work of
Wang, Qui, Yin, Qiao, and Gao \cite[Theorem 1.1]{WQYQG} and Yang and Zeng
\cite[Theorem 1.2]{yangzeng} implies that the boundary between the white, blue, and
black regions shown in Figure \ref{bq_vm1_figure} is connected (see Lemma \ref{LEM:BASIN_BOUNDARY}).
Furthermore, because ${\mathcal A}_q(-1) = {\mathcal M}$, parameters where the dynamics of $F_q(v)$
bifurcates will be in ${\mathcal B}_q(-1)$.  Again, this is special to the
case $v=-1$.

We comment on several general properties. For this and the other region
diagrams studied here, with both $v_0 \in [-1,0)$ (antiferromagnetic range) and
$v_0 > 0$ (ferromagnetic range), the outer part of the diagram, extending
infinitely far from the origin, is characterized by the property that $\lim_{m
\to \infty} F_q^m(v_0)=0$, as indicated by the white color.  For real $v_0$,
this can be understood as follows.  In both the antiferromagnetic and
ferromagnetic Potts model, an increase in $q$ introduces more fluctuations in
the system, since the spin at each site can take on values in a larger set.
Hence, for a given temperature and hence a given value of $v_0$, the system
will be in the disordered phase with no long-range spin-spin ordering.  The RG
transformation will thus move the system toward the infinite-temperature fixed
point at $v=0$.

The inner part of the diagram is comprised of blue, black, and also white
regions (the white regions being separated from the outer white region by parts
of ${\mathcal B}_q(v_0)$. An outer boundary separates the outer white region
from this complex inner set of regions.  

Let us now discuss some properties of how ${\mathcal B}_q(-1)$ intersects the
real $q$ axis.  We refer the reader to Figure \ref{bq_vm1_figure} throughout
the discussion.  We claim that:

\begin{itemize}

\item[(i)] The left-most real point where ${\mathcal B}_q(-1)$ intersects the
  real $q$ axis is $q=0$,

\item[(ii)] The right-most real point where ${\mathcal B}_q(-1)$ intersects the
  real $q$ axis is $q=3$, and

\item[(iii)] There is an infinite sequence of points $q_n$ where ${\mathcal B}_q(-1)$ intersects the real $q$ axis
converging to $q_\infty = \frac{32}{27} \in {\mathcal B}_q(-1)$.

\end{itemize}
We will first give a physical description and interpretation of these
properties.  They will later be proved as part of Theorem \ref{THM:AFM}.

Let us denote by $q_c(D_\infty) = 3$ the right most place where ${\mathcal
B}_q(-1)$ intersects the real $q$ axis.  It can be explained by the fact that
there is a qualitative change in the locus ${\mathcal B}_v(q)$ at $q=3$, namely
the appearance of an antiferromagnetic transition at $T=0$ \cite{bambihu89}.
Therefore, the intuitive approach that $(q_0,v_0) \in {\mathcal B}_q(v_0)$ if
and only if $(q_0,v_0) \in {\mathcal B}_v(q_0)$ indicates that this should lead
to $q = 3 \in \mathcal{B}_q(-1)$.  This expresses the property that the $q=3$
Potts antiferromagnet has a zero-temperature critical point on $D_\infty$. 

Note that although ${\mathcal B}_q(-1)$ crosses the real $q$ axis at this
point, the chromatic polynomial $P(D_m,3)$ has the nonzero value (\ref{pdmq3})
at $q=3$.  We remark that $q_c(D_\infty)$ is equal to the value
$q_c(S_{\infty})=3$ that was inferred in a similar manner for the Sierpinski
gasket fractal $S_\infty$ in \cite{sg} and is also the same as the value
$q_c(sq)=3$ for the (infinite) square lattice \cite{lenard}.  Interestingly, it
is also the same as the value of $q_c$ that was derived for infinite-length
self-dual strips of the square lattice~\cite{sdg}.

The nature of the Julia set $J(q)$ and hence of ${\mathcal B}_v(q)$
changes qualitatively depending on whether the discriminant (\ref{disc3}) is
positive, negative or zero, and the demarcation point between positive and
negative values occurs at the special value $q=32/27$. Connected with this, we
infer that ${\mathcal B}_q$ crosses the real $q$ axis at $q=32/27$ and,
furthermore, that this is the minimal positive value of $q$ where such a
crossing occurs. Note that the point $q=32/27$ itself is not a chromatic zero
of $D_m$ for any $m$.  More generally, it has been proved that for an arbitrary
graph, the real interval $(1,32/27]$ is free of chromatic zeros \cite{jackson}
(see also \cite{thomassen}).

Let us now start at the right-most crossing at $q=3$ and move to the left
through the adjacent inner blue region. This blue region extends down to a 
point that is the unique real root of the cubic equation 
\beq
q^3-5q^2+11q-9=0 \ ,
\label{q1cubic}
\eeq
namely, 
\beqs
q_1 &=& -\frac{1}{3}\Big ( 1 + 3\sqrt{57} \, \Big )^{1/3} +
\frac{8}{3\Big ( 1 + 3\sqrt{57} \, \Big )^{1/3}} + \frac{5}{3} \cr\cr
&=& 1.638896919...
\label{q1}
\eeqs
where ${\mathcal B}_q(-1)$ crosses the real $q$ axis. Indeed, one finds that
\begin{align*}
F_q(-1) &= v_{c,{\rm PM-FM}}(q) \qquad \mbox{if $q=q_1$ or $q=3$, and} \\
F_q(-1) &> v_{c,{\rm PM-FM}}(q) \qquad \mbox{if $q_1 < q < 3$.}
\end{align*}
The former implies that the marked
point $v=-1$ is active for $q=q_1$ and for $q=3$ so that
Lemma~\ref{LEM:ACTIVE_SUBSET_BQ} gives that they are in $\mathcal{B}_q(-1)$.
The latter implies that $(q_1,3) \in \mathcal{P}_q^\infty(-1)$ (blue region);
see Lemma \ref{LEM:FM_DYNAMICS}(ii).
Eqn.\ (\ref{q1}) was obtained by solving
\begin{align}\label{EQN:PREFIXED}
F_q^2(-1) = F_q(-1),
\end{align}
whose solutions correspond to all values of $q$ for which $v=-1$ is either a
fixed point or mapped to a fixed point.  Eqn.\ (\ref{EQN:PREFIXED}) has four
real solutions $q=1, q=\frac{3}{2}, q=q_1 \simeq 1.6388,$ and $q=3$.  (The
first two solutions correspond to passive behaviors for the marked point
$v=-1$, with it being a superattracting fixed point when $q=1$ and it being
mapped by $\ F_q$ to the superattracting fixed point at $v=0$ when
$q=\frac{3}{2}$.)

Fig. \ref{bq_vm1_figure} also shows a succession of regions and associated
crossings $q_n$ of ${\mathcal B}_q(-1)$ with the real $q$ axis, extending to
the left of $q_1$ and converging on $q_\infty=32/27$ from above.  Between these
crossings are an infinite set of regions, alternating between blue and white,
also decreasing to $q_\infty=32/27$ from above.  However, the figure,
calculated and presented to finite resolution, can only show a finite subset of
these.

We have numerically computed the ten largest crossing points of $B_q(-1)$ with
the real $q$ axis and we present them in Table \ref{crossings}.  They
were computed using a method similar to the computation of $q_1$, described
above.  More specifically, for $1 \leq k \leq 5$ we numerically solved the
equation
\begin{align*}
F_q^k(-1) = v_{c,{\rm PM-FM}}(q)
\end{align*}
for $q$ within suitably chosen intervals that were deduced from Figure
\ref{bq_vm1_figure}.

\begin{table}
\begin{center}
\begin{tabular}{|cl|}
\hline
$q_0 =$ &  $3$  \\
$q_1 =$ &  $1.6388969195$  \\
$q_2 =$ &  $1.4097005138$  \\
$q_3 =$ &  $1.3232009243$  \\
$q_4 =$ &  $1.2798668287$  \\
$q_5 =$ &  $1.2546493642$  \\
$q_6 =$ &  $1.2385319865$  \\
$q_7 =$ &  $1.2275429153$  \\
$q_8 =$ &  $1.2196860382$  \\
$q_9 =$ &  $1.2138598416$  \\
\hline
\end{tabular}
\vspace{0.1in}
\caption{\label{crossings} Approximate values for the $10$ largest
  intersections between $B_q(-1)$ and the real $q$ axis.}
\end{center}
\end{table}

It is of interest to compare the exact results discussed above and depicted on
the left hand side of in Fig.  \ref{bq_vm1_figure} with the chromatic zeros of
$D_m$ calculated for finite $m$ shown on the right hand side of
Fig.~\ref{bq_vm1_figure}.  Extensive experience with chromatic zeros of
sections of regular lattices has shown that a subset of these approach the
locus ${\mathcal B}_q(-1)$ as the number of vertices gets large (e.g.,
\cite{w}, \cite{pg}-\cite{k}).  Here we observe a similar behavior, although
for the diamond hierarchical lattice ${\mathcal B}_q(-1)$ is obviously much
more complicated than the real algebraic curves and possible line segments
comprising the loci ${\mathcal B}_q(-1)$ for the $n \to \infty$ limits of
sections of regular lattices and related chain graphs. In particular, one can
see (complex-conjugate pairs of) zeros near $q=0$ and $q=3$, as well as a
clustering of chromatic zeros forming a wedge-shaped pattern, with the apex of
the wedge facing left and located on the real $q$ axis at $q \simeq 1.2$, close
to $q_\infty$.  The zeros that we have calculated for $D_m$ graphs show
considerable scatter, and in this respect they differ from the chromatic zeros
that were calculated in \cite{sg} for Sierpinski graphs. From a comparison of
the zeros for $1 \le m \le 4$, we find that the left-most complex-conjugate
pair of zeros move toward the point $q=0$ as $m$ increases, in agreement with
the property deduced from the analysis leading to Fig. \ref{bq_vm1_figure},
that $q=0$ is a crossing of ${\mathcal B}_q(-1)$.  This is consistent with the
fact that in other cases (e.g. \cite{wcyl,nec,a,s3a,ta} where this behavior (of
complex-conjugate pairs of zeros in the vicinity of $q=0$ moving toward the
latter point as $n$ increases) is observed, and one has exact results for
${\mathcal B}_q(-1)$, it is associated with the property that for $n \to
\infty$, the locus ${\mathcal B}_q(-1)$ passes through $q=0$ and separates the
complex $q$ plane into different regions.


In analyzing these chromatic zeros and their limiting behavior for $m \to
\infty$, it is useful to recall some rigorous results on zero-free regions on
the real $q$ axis. Since the signs of descending powers of $q$ in $P(G,q)$
alternate, an elementary property is that $P(G,q)$ has no zeros in the interval
$(-\infty,0)$.  For an arbitrary graph $G$, there are also no chromatic zeros
in the interval (0,1) and, as mentioned above, in the interval $(1,32/27]$; see
\cite{jackson}, \cite{thomassen}, \cite{dong}.  Thus, although ${\mathcal
  B}_q(-1)$ crosses the real $q$ axis at the point $q=32/27$, this point itself
is not a chromatic zero.  Since $P(G,q)$ always has a factor of $q$, it always
vanishes at $q=0$, and if, as is the case here, $G$ has at least one edge, then
$P(G,q)$ also vanishes at $q=1$.  For $1 \le m \le 4$, we find that the only
real zeros of $P(D_m,q)$ are $q=0, \ 1$. It is interesting to note that
although $P(D_m,2)=2$, this polynomial can be quite small for part of the
interval $1 \le q \le 2$. For example, as $q$ increases from 1 to 2, $P(D_2,q)$
reaches a local maximum of 0.041 at $q=1.1$, then decreases to a local minimum
of 0.0080 at $q=1.36$, and finally increases to 2 as $q \to 2$. In the same
interval, $P(D_3,q)$ reaches a maximum of 0.0090 at $q=1.02$, decreases to a
minimum of approximately $1.7 \times 10^{-9}$ at $q=1.37$, and then increases
to 2 at $q=2$.

Let us remark on another property of $P(D_m,q)$.  Because $D_m$ is
bipartite, $P(D_m,2)=2$. By explicit iterative calculation, we obtain
\beq
P(D_m,3) = 2 \cdot 3^{n(D_m)/2} \ . 
\label{pdmq3}
\eeq
Consequently, 
\beq
W(D_\infty,3) = \sqrt{3} \ ,
\label{wdhlq3}
\eeq
so that the 3-state Potts AFM has ground-state entropy
on the $D_\infty$ fractal given by 
\beq
S_0(D_\infty) = \frac{k_B}{2}\ln 3 \ .
\label{dhlentropy}
\eeq
Interestingly, these results are the same as for the infinite-length
square-lattice ladder graph (with any longitudinal BC) \cite{w}), for which
\beq
W(sq,2 \times \infty,q)=\sqrt{q^2-3q+3} \ , 
\label{wlad}
\eeq
and hence $W(sq,2 \times \infty,3)=\sqrt{3}$. The cyclic and M\"obius strips
of the square lattice with width $L_y=2$ vertices are $\Delta$-regular graphs
with vertex degree $\Delta=3$, and the free $L_y=2$ square-lattice strip also
has $\Delta_{eff}=3$ in the $m \to \infty$ limit.  These values of $\Delta$ and
$\Delta_{eff}$ are the same as the value~(\ref{delta_dhl}). 

  
\section{Zeros in the $q$ Plane for the Potts Antiferromagnet at Nonzero
  Temperature }
\label{qzerosafm}

We next consider the zeros of $Z(D_m,q,v_0)$ for the Potts antiferromagnet with
temperature $T_0 > 0$ which corresponds to the range $-1 < v_0 \le 0$.  On the
left sides of Figures \ref{bq_vm0p8_figure} - \ref{bq_vm0p2_figure} we present
computer-generated images of the regions
\begin{align*}
{\mathcal P}_q^0(v_0) &:= \{q \in \mathbb{C} \setminus \{0\} \, : \, F_q^m(v_0) \rightarrow 0\} \qquad \mbox{colored white, and} \\
{\mathcal P}_q^\infty(v_0) &:= \{q \in \mathbb{C} \setminus \{0\} \, : \, F_q^m(v_0) \rightarrow \infty \} \qquad \mbox{colored blue}
\end{align*}
for $v_0 = -0.8, -0.5$, and $-0.2$.  Any point that is not in ${\mathcal
  P}_q^0(v_0)$ or ${\mathcal P}_q^\infty(v_0)$ is colored black.  According to
Proposition \ref{PROP_DESCRIBING_FIGURES}, ${\mathcal B}_q(v_0)$ contains any
point of the boundary between the white, blue, and black sets.  It may contain
additional points, but they are necessarily in the interior of the set of black
points.

For comparison, on the right sides of Figures \ref{bq_vm0p8_figure}
- \ref{bq_vm0p2_figure} we present the numerically computed zeros of the
reduced partition function $Z_r(D_4,q,v_0)$ at these values of $v_0$.  (As
usual, the zero at $q=0$ is omitted.)


\begin{figure}
\begin{center}
\scalebox{1.2}{
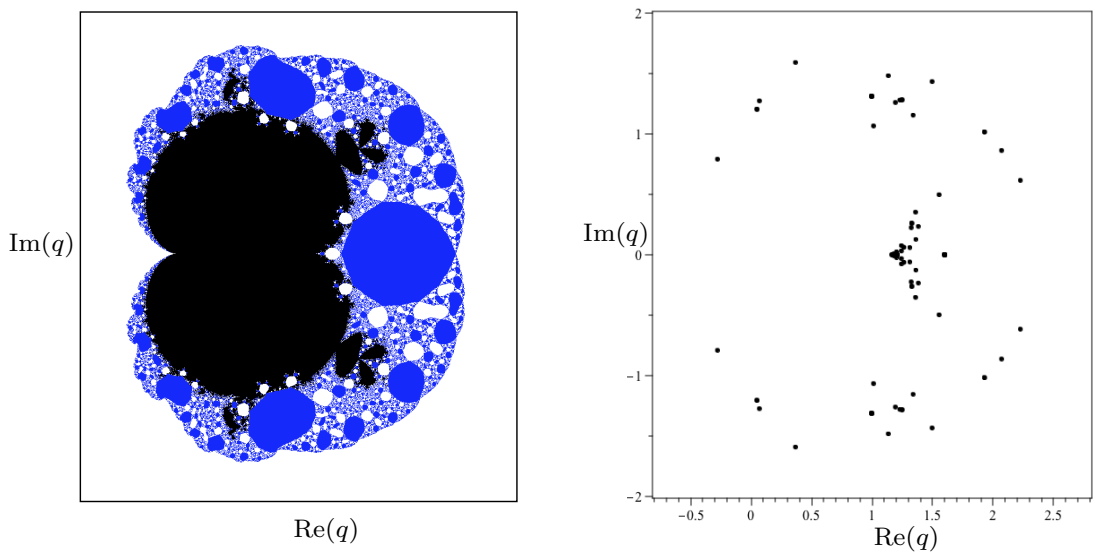
}
\end{center}
\caption{ Left: Region diagram for $D_\infty$ in the complex $q$ plane for
  $v=-0.8$.  Right: Zeros of the reduced partition function $Z_r(D_4,q,-0.8)$
  (171 zeros).  Both left and right figures depict $-0.8 < {\rm Re}(q) < 2.8$
  and $-2 < {\rm Im}(q) < 2$.  }
\label{bq_vm0p8_figure}
\end{figure}


\begin{figure}
\begin{center}
\scalebox{1.2}{
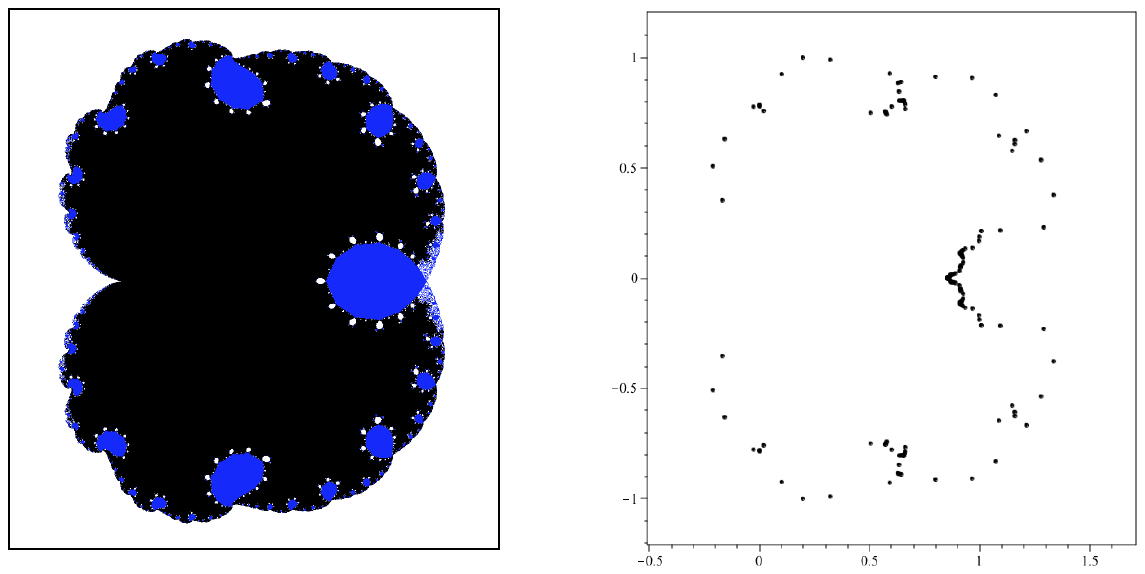
}
\end{center}
\caption{
Left: Region diagram for $D_\infty$ in the complex $q$ plane for $v=-0.5$.  
Right: Zeros of the reduced partition function $Z_r(D_4,q,-0.5)$ (171 zeros).
    Both left and right figures depict 
    $-0.5 < {\rm Re}(q) < 1.7$ and $-1.2 < {\rm Im}(q) < 1.2$.
}

\label{bq_vm0p5_figure}
\end{figure}


\begin{figure}
\begin{center}
\scalebox{1.2}{
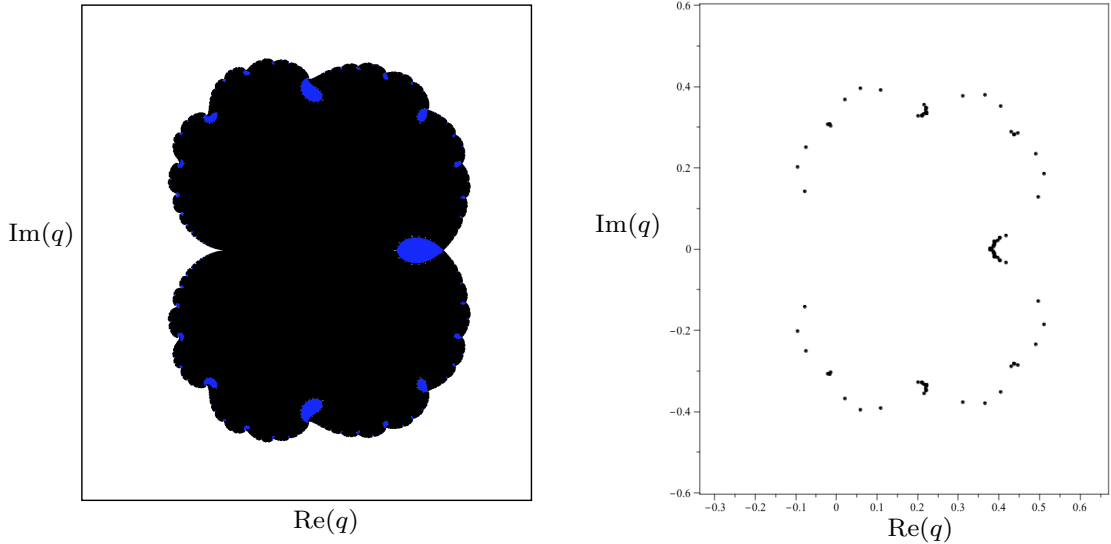
}
\end{center}
\caption{
Left: Region diagram for $D_\infty$ in the complex $q$ plane for $v=-0.2$.  
Right: Zeros of the reduced partition function $Z_r(D_4,q,-0.2)$ (171 zeros).
    Both left and right figures depict 
    $-1/3 < {\rm Re}(q) < 2/3$ and $-0.6 < {\rm Im}(q) < 0.6$.
}
\label{bq_vm0p2_figure}
\end{figure}
%


First, we observe that in the case of infinite temperature, or equivalently,
zero spin-spin coupling, $v_0=0$ we have that $Z(G,q,0)=q^{n(G)}$ for any graph
$G$, so that all of the zeros of $Z(G,q,0)$ occur at $q=0$.  Previous studies
(e.g., \cite{a,ta}) showed that generically, as $v_0$ approaches 0, the zeros
of $Z(G,q,v_0)$ in the $q$ plane progressively move in toward the origin,
$q=0$.  We see that behavior here for the diamond hierarchical graphs $D_m$, as
evident from the decreasing scale of Figures \ref{bq_vm1_figure} and Figures
\ref{bq_vm0p8_figure}-\ref{bq_vm0p2_figure} as $v_0$ increases from $-1$ to
$-0.2$.

Note also that for $v_0 = -0.8, -0.5$ and $-0.2$ the zeros of $Z(D_4,q,v_0)$
show somewhat less scatter than at $v_0=-1$.  Moreover, one sees that the
wedge-like formation of zeros moves to the left as $v_0$ increases,
intersecting the real real axis at $q \simeq 1.2$ for $v_0 = -0.8$, $q \simeq
0.85$ for $v_0=-0.5$, and at $q \simeq 0.38$ for $v_0=-0.2$.

As $v$ increases in the interval $v_0 \in (-1,0]$, there is an evident
simplification in the different regions on the left side of Figures
\ref{bq_vm0p8_figure}-\ref{bq_vm0p2_figure}, as compared with the $v_0=-1$
case.  As was found for regular lattice graphs (e.g., \cite{a}), as $v_0$
increases in the interval $(-1,0]$, the maximal point, $q_c(v_0)$, at which
${\mathcal B}_q(v_0)$ crosses the real axis decreases. This behavior is
expected, since as $v_0 \to 0$, all of the zeros move in toward $q=0$. The
crossing at $q=0$ remains present for all $v_0 \in (-1,0]$.

Although it is not completely evident from Figures \ref{bq_vm0p8_figure} -
\ref{bq_vm0p2_figure}, for any $-1 \leq v_0 < 0$ the locus ${\mathcal
  B}_q(v_0)$ continues to intersect the real $q$ axis in infinitely many
points, just like in the case $v_0=-1$.  One can observe this by using the
computer to zoom in when investigating the region diagrams.  However, we will
prove it rigorously in the theorem below.

\vspace{0.1in}

The following theorem summarizes the antiferromagnetic case:

\vbox{
\begin{theorem}\label{THM:AFM}
For any $-1 \leq v_0 < 0$ 
the locus ${\mathcal B}_q(v_0)$ intersects the axis at $0$ and at
\begin{align*}
q_c(v_0)=\left( -2-\sqrt {-v_0} \right) v_0 > 0.
\end{align*}
The locus ${\mathcal B}_q(v_0)$ does not intersect the real $q$ axis at any point outside of the interval $[0,q_c(v_0)]$.

Furthermore, there is an infinite sequence of points $q_k(v_0)$ where 
${\mathcal B}_q(v_0)$ intersects the real $q$ axis converging to 
\begin{align}\label{EQN:QINFINITY}
q_\infty(v_0) =  
\begin{cases}
\frac{32}{27} \qquad & \mbox{if $-1 \leq v_0 \leq -\frac{8}{9}$} \\
\left( -1-\sqrt {1+v_0} \right) v_0 \qquad  & \mbox{if $-\frac{8}{9} \leq v_0 < 0$}.
\end{cases}
\end{align}
which is also in ${\mathcal B}_q(v_0)$.
\end{theorem}
}
\noindent
We refer the reader to Figure \ref{FIG_Q_C} for a plot of $q_c(v_0)$.  The
values of $q_c$ for these values of $v_0$ shown in Figures \ref{bq_vm1_figure},  \ref{bq_vm0p8_figure}, \ref{bq_vm0p5_figure}, and \ref{bq_vm0p2_figure}
are $q_c(-1) = 3, \, q_c(-0.8) \simeq 2.316, \, q_c(-0.5) \simeq
1.354$, and $q_c(-0.2) \simeq 0.4894$, respectively.  Theorem \ref{THM:AFM}
will be proved in Section \ref{SEC:PROOFS}.

\begin{figure}
\includegraphics[scale=0.4]{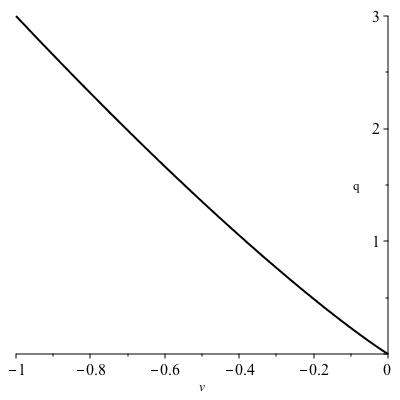}
\caption{Plot of the upper intersection point $q_c(v_0)$ between ${\mathcal B}_q(v_0)$ and the real $q$-axis
for $-1 \leq v_0 < 0$.  See Theorem \ref{THM:AFM}.
\label{FIG_Q_C}}
\end{figure}


\section{Zeros in the $q$ Plane for the Ferromagnetic Potts Model}
\label{qzerosfm}

We now present some results on the zeros of $Z(D_m,q,v)$ and the locus
${\mathcal B}_q(v_0)$ in the complex $q$ plane for the Potts ferromagnet,
corresponding to $v_0 \ge 0$.  On the left sides of Figures \ref{bq_v1_figure}
- \ref{bq_v99_figure} we present computer-generated images of the regions
\begin{align*}
{\mathcal P}_q^0(v_0) &:= \{q \in \mathbb{C} \setminus \{0\} \, : \, F_q^m(v_0) \rightarrow 0\} \qquad \mbox{colored white, and} \\
{\mathcal P}_q^\infty(v_0) &:= \{q \in \mathbb{C} \setminus \{0\} \, : \, F_q^m(v_0) \rightarrow \infty \} \qquad \mbox{colored blue}
\end{align*}
for $v_0 = 1, 2, 4$, and $99$.  As in previous figures, any point that is not
in ${\mathcal P}_q^0(v_0)$ or ${\mathcal P}_q^\infty(v_0)$ is colored black.
According to Proposition \ref{PROP_DESCRIBING_FIGURES}, ${\mathcal B}_q(v_0)$
contains any point of the boundary between the white, blue, and black sets.  It
may contain additional points, but they are necessarily in the interior of the
set of black points.

For comparison, on the right sides of Figures \ref{bq_v1_figure} -
\ref{bq_v99_figure} we present the numerically computed zeros of the 
partition function $Z(D_4,q,v_0)$ at these values of $v_0$.  (As usual, the
zero at $q=0$ is omitted.)


\begin{figure}
\begin{center}
\scalebox{1.2}{
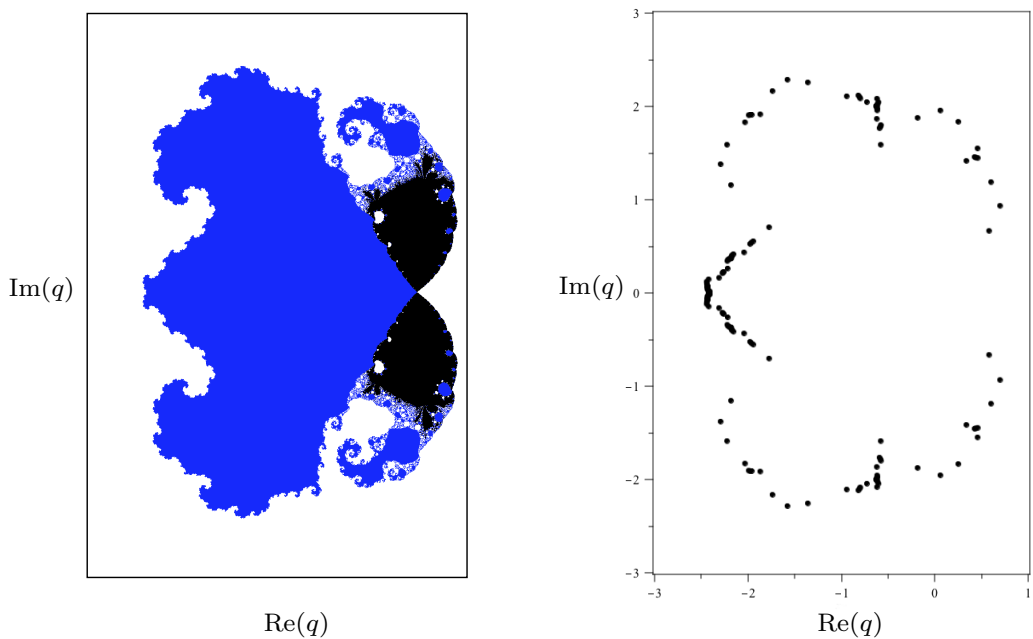
}
\end{center}
\caption{ Left: Region diagram for $D_\infty$ in the complex $q$ plane for
  $v_0=1$.  Right: Zeros of the reduced partition function $Z_r(D_4,q,1)$ (171
  zeros).  Both left and right figures depict $-3 < {\rm Re}(q) < 1$ and $-3 <
  {\rm Im}(q) < 3$.  }
\label{bq_v1_figure}
\end{figure}

\begin{figure}
\begin{center}
\scalebox{1.2}{
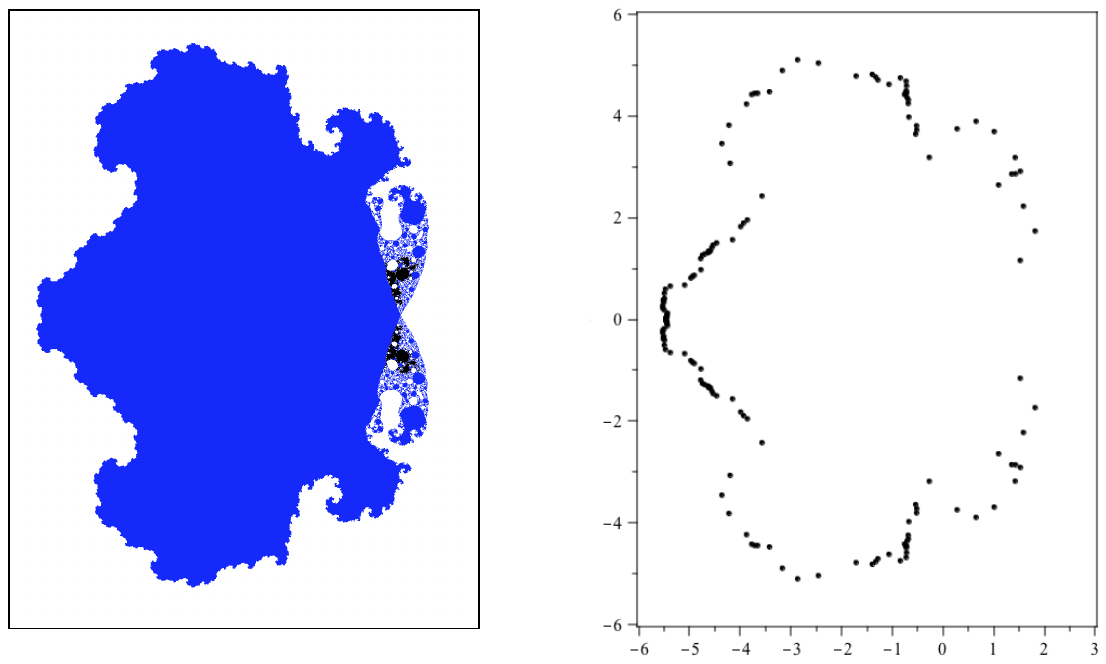
}
\end{center}
\caption{ Left: Region diagram for $D_\infty$ in the complex $q$ plane for
  $v_0=2$.  Right: Zeros of the reduced partition function $Z_r(D_4,q,2)$ (171
  zeros).  Both left and right figures depict $-6 < {\rm Re}(q) < 3$ and $-6 <
  {\rm Im}(q) < 6$.  }
\label{bq_v2_figure}
\end{figure}


\begin{figure}
\begin{center}
\scalebox{1.2}{
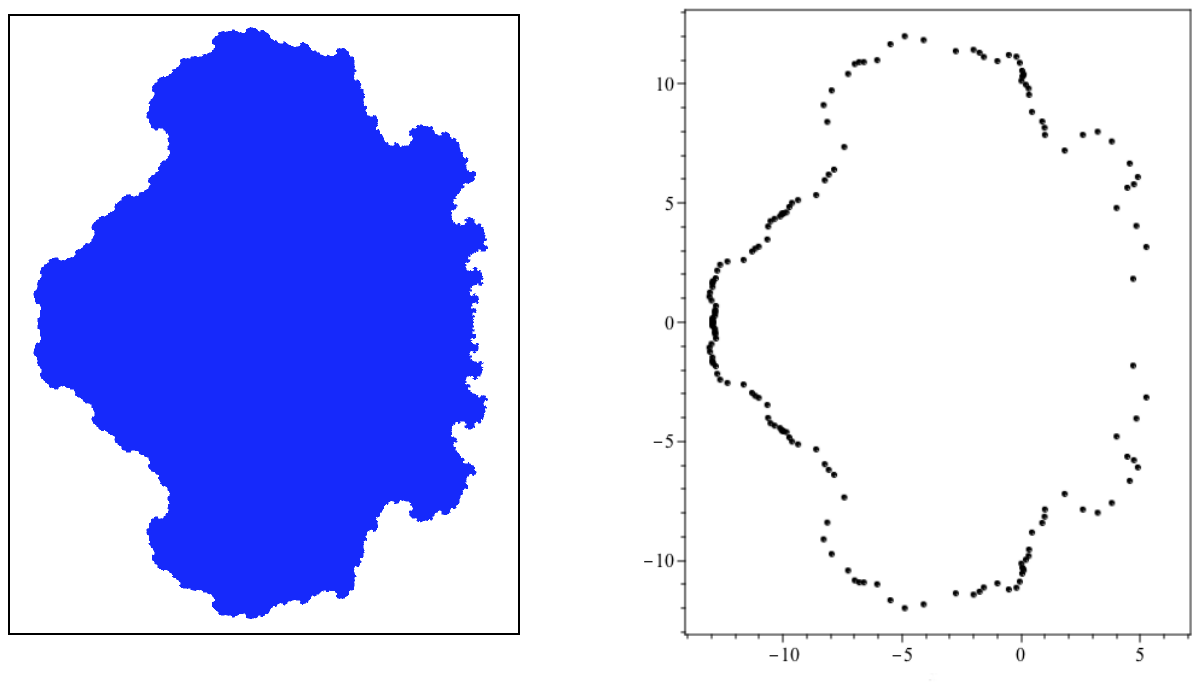
}
\end{center}
\caption{
Left: Region diagram for $D_\infty$ in the complex $q$ plane for $v_0=4$.
Right: Zeros of the reduced partition function $Z_r(D_4,q,4)$ (171 zeros).
    Both left and right figures depict
    $-14 < {\rm Re}(q) < 7$ and $-13 < {\rm Im}(q) < 13$.
}
\label{bq_v4_figure}
\end{figure}


\begin{figure}
\begin{center}
\scalebox{1.2}{
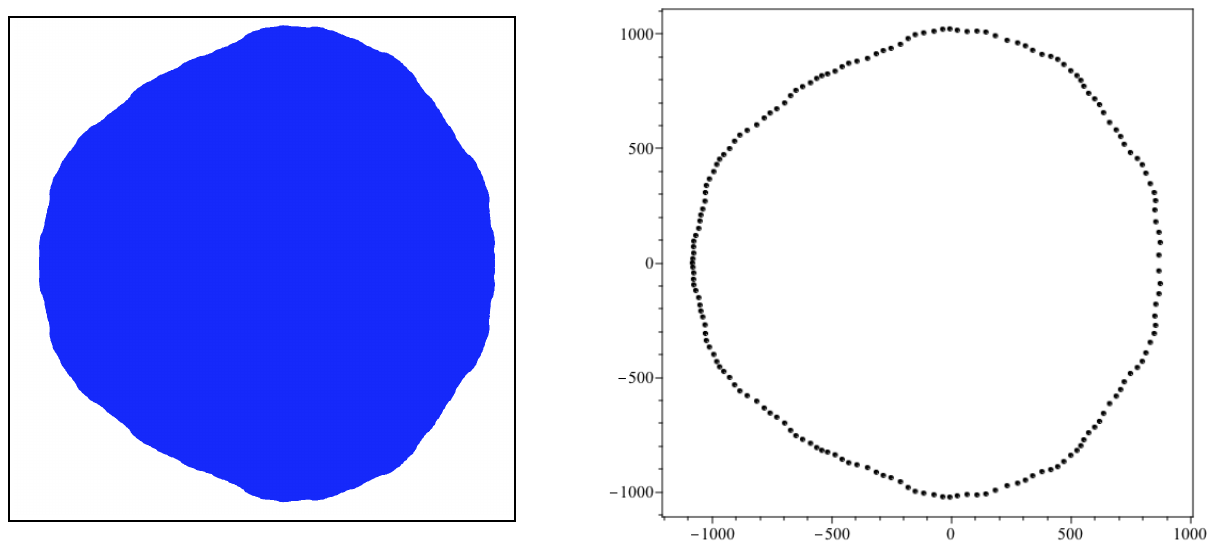
}
\end{center}
\caption{
Left: Region diagram for $D_\infty$ in the complex $q$ plane for $v_0=99$.
Right: Zeros of the reduced partition function $Z_r(D_4,q,99)$ (171 zeros).
    Both left and right figures depict
    $-1200 < {\rm Re}(q) < 1000$ and $-1100 < {\rm Im}(q) < 1100$.
}
\label{bq_v99_figure}
\end{figure}


In general, for all four of these ferromagnetic values of $v$, especially for
the largest two values, the region diagram appears considerably simpler than
for the antiferromagnetic values. As in the antiferromagnetic case, the region
diagrams have an outer white area extending infinitely far from the origin,
separated from an inner blue and black portion by part of ${\mathcal
  B}_q(v_0)$.  For $v_0=1$ and $v_0 = 2$ there are still some black regions,
but they have become rather small, and for $v_0=4$ and $v_0=99$, at the
resolution of the figures, one sees only an inner blue region and an outer
white region, separated by ${\mathcal B}_q(v_0)$.  Furthermore, they appears to
become smoother, approaching a nearly circular form for large $v_0$, as will be
discussed further below. The computer images also indicate that the Hausdorff
dimension of ${\mathcal B}_q(v_0)$ decreases as $v_0$ increases sufficiently.

Another difference from the antiferromagnetic values of $-1 \leq v_0 < 0$ is
that for the ferromagnetic values $v_0 > 0$ the locus $\mathcal{B}_q(v_0)$
intersects the real $q$ axis in only two points:

\begin{theorem}\label{THM:FM}
For any $v_0 > 0$ the locus $\mathcal{B}_q(v_0)$ intersects the real $q$-axis only at the two points $q_-(v_0) < q_+(v_0)$ given
by
\begin{align*}
q_\pm(v_0) = \left( -1 \pm \sqrt {1+v_0} \right) v_0.
\end{align*}
\end{theorem}
\noindent
We refer the reader to Figure \ref{FIG_Q_PLUS_MINUS} for a plot of
$q_\pm(v_0)$.  The values of $q_-$ and $q_+$ for these values of $v_0$ shown in
Figures \ref{bq_v1_figure} - \ref{bq_v99_figure} are $(-2.414,0.4142)$,
$(-5.464,1.464)$, $(-12.944,4.944)$, and $(-1089,891)$, respectively.  Theorem
\ref{THM:FM} will be proved in Section \ref{SEC:PROOFS}.

\begin{figure}
\includegraphics[scale=0.4]{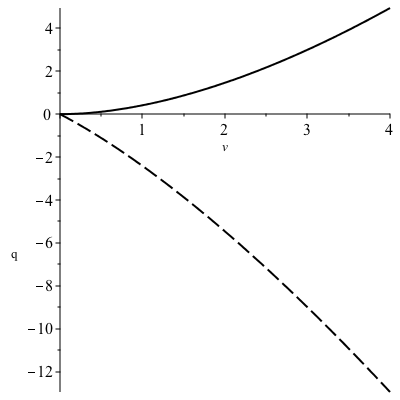}
\caption{Plot of the lower $q_-(v_0)$ (dashed) and upper $q_+(v_0)$ (solid) intersection points between ${\mathcal B}_q(v_0)$ and the real $q$-axis
for $v_0 > 0$.  See Theorem \ref{THM:FM}.
\label{FIG_Q_PLUS_MINUS}}
\end{figure}

We can give a statistical physics explanation of the crossing of ${\mathcal B}_q(v_0)$
on the positive real $q$ axis; the Potts ferromagnet has a phase transition at
this value of $q=q_c$ for the given value of the temperature variable $v_0$,
namely $T_{c,PM-FM}$ as given in Eq.\ (\ref{tc_pmfm}).  For larger $q$, the
system is more disordered; for integral $q$, this can be understood from the
fact that each spin on the lattice can take values in a larger set,
$\{1,...,q\}$.  In this case, the infinite iteration of the RG transformation
maps the temperature to $T=\infty$, or equivalently, decreases $v$ to 0,
hence the white color.  For $0 < q < q_c$ the system is more ordered,
so the infinite iteration of the RG transformation maps the temperature
variable to $T=0$, i.e., $v=\infty$, hence the blue color. 

In contrast with the scattered pattern of zeros in the $q$ plane for the $T=0$
AFM case $v_0=-1$, the zeros in the $q$ plane for the FM case illustrated by
these four values of $v_0$ tend to cluster along a curve.  This curve encircles
the origin. As $v_0$ gets large, this curve assumes an oval-like shape.  In
\cite{lq} (see also the related \cite{hungwu}), we showed that, in the $n \to
\infty$ limit of a recursive family of graphs, as $|v_0|$ increases to values
$\gg 1$, the accumulation set of zeros in the $q$ plane, ${\mathcal B}_q(v_0)$,
forms a closed oval curve that encircles the origin and crosses the positive
$q$ axis at
\beq
q \simeq |v|^{\Delta_{eff}/2} \ , 
\label{qc_largey}
\eeq
where $\Delta_{eff}$ was defined in Eq.\ (\ref{delta_eff}). 
In this $m \to \infty$ limit, ${\mathcal B}_q(v_0)$ forms a closed oval curve
approaching a circle, crossing the positive and negative real $q$ axes at a
value of $|q|$ that behaves asymptotically as
\beq
|q| \sim |v_0|^{3/2} \ . 
\label{qcasymptotic}
\eeq

Our calculations of zeros of $Z(D_m,q,v_0)$ for finite $m$ are in agreement
with this result, as is illustrated in Fig. \ref{bq_v99_figure}.  Note that for
$m=4$ we have
\beq
\Delta_{eff} = \frac{3}{1+2^{-7}} = 2.98 \quad {\rm for} \ D_4 \ , 
\label{delta_dm4}
\eeq
which is quite close to the limiting value $\lim_{m \to \infty} \Delta_{eff}=3$
in Eq.\ (\ref{delta_dhl}).


\section{Proofs of Theorems \ref{THM:AFM} and \ref{THM:FM}.}
\label{SEC:PROOFS}

We will consider the dynamics of the RG mapping $F_q(v)$ for $v \in \mathbb{R}$ and $q \in \mathbb{R} \setminus \{0\}$.

\subsection{Lemmas and setup.}
Let us briefly summarize some properties of $F_q(v)$ that will be used throughout this section.

For any  $q \in \mathbb{R} \setminus \{0\}$ the mapping $F_q(v)$ has a superattracting fixed point at $v=0$
and another fixed point $v_{c,{\rm PM-FM}}(q) > 0$.  Moreover, 
$v_{c,{\rm PM-FM}}(q)$ is the unique positive fixed point of $F_q(v)$ for any $q \in \mathbb{R} \setminus \{0\}$, a property that will
play an important role in several of the proofs.

For $0 < q < \frac{32}{27}$ the mapping $F_q(v)$ has two additional fixed points $v_-(q)$ and $v_+(q)$ with
\begin{align*}
v_-(q) < v_+(q) < 0 < v_{c,{\rm PM-FM}}(q).
\end{align*}
When $q =  \frac{32}{27}$ these two fixed points collide: $v_-(q) = v_+(q)= -\frac{8}{9}$.
Figure \ref{FIG_FIXED_POINTS} shows a plot of how $v_-(q)$, $v_+(q)$, and $v_{c,{\rm PM-FM}}(q)$ depend on $q$.

The mapping $F_q(v)$  has a pole at 
\begin{align*}
v = -q/2 \qquad \mbox{where} \qquad \lim_{v \rightarrow -q/2} F_q(v) = + \infty.
\end{align*}
The critical points of $F_q(v)$ are:
\begin{align*}
v = 0, \quad -q, \quad -1 \pm \sqrt{1-q}.
\end{align*}
Note that
\begin{align*}
F_q( -1 \pm \sqrt{1-q}) = -1
\end{align*}
for any $q$.

\begin{lemma}\label{LEM:INVARIANT}
The extended real interval $[-1,\infty]$ is invariant under $F_q(v)$.
\end{lemma}

\begin{proof}
One can check that the expression for $F_q(v) + 1$ is a perfect square.
\end{proof}

\begin{lemma}\label{LEM:REPELLING_FP}
  For any real $q \neq 0$ we have that $v_{c,{\rm PM-FM}}(q)$ is a repelling
  fixed point for~$F_q(v)$ satisfying $F'_q(v_{c,{\rm PM-FM}}(q)) > 1$.
\end{lemma}

\begin{proof}
  We remark that for $q \neq 0$ the fixed point $v_{c,{\rm PM-FM}}(q)$ is a
  solution to the equation $F_q(v) = v$ that occurs with multiplicity one so
  that we cannot have $F_q'(v_{c,{\rm PM-FM}}(q)) = 1$.

  {\bf Case 1:} $\bm{q > 0}$.  
  The pole and all of the critical points of
  $F_q(v)$ lie in $(-\infty,0]$.  A simple calculation shows that $F_q'(v) > 0$
  for $v > 0$.  Since $F_q'(0) = 0$ and $v_{c,{\rm PM-FM}}(q)$ is the unique
  fixed point of $F_q(v)$ occurring at positive $v$ we conclude that $F_q(v) <
  v$ for $0 < v < v_{c,{\rm PM-FM}}(q)$ and that $F_q(v) > v$ for $v >
  v_{c,{\rm PM-FM}}(q)$.  Since $F_q'(v_{c,{\rm PM-FM}}(q)) \neq 1$ we must
  have $F_q'(v_{c,{\rm PM-FM}}(q)) > 1$, as desired.

{\bf Case 2:} $\bm{q < 0}$. The pole at $-q/2$ and two of the critical points $-q$ and $-1 + \sqrt{1-q}$ occur at positive $v$ and
they do so with the following order:
\begin{align*}
-1 + \sqrt{1-q} < -q/2 < -q.
\end{align*}
Since $F_q(-1 + \sqrt{1-q}) = -1 < -1 + \sqrt{1-q}$ and $\lim_{v \rightarrow -q/2} F_q(v) = \infty$ the Intermediate Value Theorem
implies that $F_q$ has a fixed point between $-1 + \sqrt{1-q}$ and $-q/2$.  Since $v_{c,{\rm PM-FM}}(q)$ is the unique positive
fixed point of $F_q$ we conclude that 
\begin{align*}
-1 + \sqrt{1-q} < v_{c,{\rm PM-FM}}(q) < -q/2.
\end{align*}
Moreover, one can check that $F_q'(v)$ is positive on $(-1 + \sqrt{1-q},-q/2)$.

Suppose for contradiction that $0 < F_q'(v_{c,{\rm PM-FM}}(q)) < 1$.  Then, for
\begin{align*}
v_{c,{\rm PM-FM}}(q) < v < -q/2,
\end{align*}
with $v$ chosen sufficiently close to $v_{c,{\rm PM-FM}}(q)$, one has $F_q(v) < v$. Since $\lim_{v \rightarrow -q/2} F_q(v) = + \infty$
the Intermediate Value Theorem would imply that there is an additional fixed point $v_\bullet$ of $F_q$ with $v_{c,{\rm PM-FM}}(q) < v_\bullet < 
-q/2$.  This contradicts that $v_{c,{\rm PM-FM}}(q)$ is the unique positive fixed point of~$F_q(v)$.
Since $F_q'(v_{c,{\rm PM-FM}}(q)) \neq 1$ we must therefore have $F_q'(v_{c,{\rm PM-FM}}(q)) > 1$.
\end{proof}

\begin{lemma}\label{LEM:PART_OF_AFM_DYNAMICS}
We have:
\begin{itemize}
\item[(i)] If $q < 0$ then for any $-1 \leq v_0 \leq 0$ we have $F_q^m(v_0) \nearrow 0$ as $m\rightarrow \infty$.
\item[(ii)] If $q > 0$ is sufficiently small then for any $-1 \leq v_0 < v_-(q)$ we have $F_q^m(v_0) \nearrow v_-(q)$ as $m\rightarrow \infty$.
\end{itemize}
\end{lemma}

\begin{proof}
Claim (i):
When $q < 0$ the pole and the critical points $-q$ and $-1 + \sqrt{1-q}$ occur at positive~$v$.  The critical point $-1 - \sqrt{1-q}$ occurs
for $v < -1$.  The only real fixed points for $F_q(v)$ are $0$ and $v_{c,{\rm PM-FM}}(q) > 0$.  A calculation
shows that $F_q(-1) > -1$  and hence that $F_q(v) > v$ for all $-1 \leq v < 0$.  Another calculation shows
that $F_q(v) < 0$ for all $-1 \leq v < 0$.  Therefore, the sequence $F_q^m(v_0)$ is increasing sequence that is bounded above by $0$.
It must converge to some limit, which will be a fixed point for $F_q(v)$.  Therefore, $F_q^m(v_0) \rightarrow 0$.

\vspace{0.1in}
Claim (ii):
Suppose $0 < q < 1$ and that $q$ is sufficiently small that $-1 \leq v_0 < v_-(q)$.  In this case, we have 
\begin{align*}
F_q(-1) = \frac{(q-1)^2}{(q-2)^2} -1 > -1 \qquad \mbox{and} \qquad F_q'(-1) = -4\,{\frac { \left( q-1 \right) ^{2}}{ \left( q-2 \right) ^{3}}} > 0. 
\end{align*}
The critical point $-1 - \sqrt{1-q}$ occurs at $v < -1$ and the other two critical points and the pole of $F_q$ are ordered as
\begin{align*}
-1 < -q < -1 + \sqrt{1-q} < -\frac{q}{2}.
\end{align*}
We have that
\begin{align*}
F_q(-q) = q^2 - 2q < -q
\end{align*}
so we conclude that $-1 < v_-(q) < -q$ and hence that $F'_q(v) > 0$ for all $-1 \leq v \leq v_-(q)$.  In particular, we have
\begin{align*}
v < F_q(v) < v_-(q) \qquad \mbox{for all} \qquad  -1 \leq v < v_-(q).
\end{align*}
This implies that for any $-1 \leq v_0 < v_-(q)$ we have $F_q^m(v_0) \nearrow v_-(q).$ 
\end{proof}

\begin{lemma}\label{LEM:FM_DYNAMICS}
For any real $q \neq 0$ we have
\begin{itemize}
\item[(i)] If $0 \leq v_0 < v_{c,{\rm PM-FM}}(q)$ then $F_q^m(v_0) \rightarrow 0$ as $m \rightarrow \infty$,
\item[(ii)] If $v_0 > v_{c,{\rm PM-FM}}(q)$ then $F_q^m(v_0) \rightarrow \infty$ as $m \rightarrow \infty$.
\end{itemize}
\end{lemma}
\noindent
We remark that in Claim (i) when $q < 0$ the orbit of $v_0$ may pass into the
interval $[-1,0)$.

\begin{proof}
We split the proof of Claim (i) into two cases:

{\bf Case 1:} $\bm{q > 0}$.
As in the first paragraph of the proof of Lemma \ref{LEM:REPELLING_FP}, we have 
\begin{align*}
0 < F_q(v) < v \quad \mbox{for} \quad 0 < v < v_{c,{\rm PM-FM}}(q).
\end{align*} 
 If $0 \leq v_0 < v_{c,{\rm PM-FM}}(q)$ then the orbit $F_q^m(v_0)$ forms a
decreasing sequence, which is bounded below by $0$, and hence converges.  The limit must be a fixed point of $F_q(v)$
and, since $v_{c,{\rm PM-FM}}(q)$ is the only positive fixed point of
$F_q(v)$, we see that $F_q^m(v_0) \rightarrow 0$.  


{\bf Case 2:} {$\bm{q < 0}$}.
As in the second paragraph of the proof of Lemma \ref{LEM:REPELLING_FP}, the critical points, pole, and fixed points
of $F_q(v)$ occur with the following order:
\begin{align}\label{EQN:ORDER}
0 < -1 + \sqrt{1-q} < v_{c,{\rm PM-FM}}(q) < -q/2 < -q.
\end{align}
We will show that $[-1,v_{c,{\rm PM-FM}}(q))$ is invariant
under $F_q(v)$ and that for any $v_0 \in [-1,v_{c,{\rm PM-FM}}(q))$ we have
$F_q^m(v_0) \rightarrow 0$.  

Let us write $[-1,v_{c,{\rm PM-FM}}(q)) = I_1 \cup I_2 \cup I_3$ with
\begin{align*}
I_1 = [-1,0], \qquad I_2 = \left[0,-2 + \sqrt {-2\,q+4}\right], \qquad \mbox{and} \qquad I_3 = \left(-2 + \sqrt {-2\,q+4},v_{c,{\rm PM-FM}}(q)\right).
\end{align*}
For $v_0 \in I_1$ it follows from Lemma
\ref{LEM:PART_OF_AFM_DYNAMICS}(i) that $F_q^m(v_0) \rightarrow 0$.

The zeros of $F_q(v)$ are
\begin{align*}
0, -2 \pm \sqrt {-2\,q+4}.
\end{align*}
Since $0 < -1 + \sqrt{1-q} < -2 + \sqrt {-2\,q+4}$ and $F_q(-1 + \sqrt{1-q}) = -1$ we have $F(I_2) \subset I_1$.
Therefore, for any $v_0 \in I_2$ we have $F_q^m(v_0) \rightarrow 0$.  

Now consider $v_0 \in I_3$.  On this interval, we have $0 \leq F_q(v) < v$.
The orbit $F_q^m(v_0)$ cannot remain in $I_3$ because, if it did, the
orbit would form a decreasing sequence that is bounded below by $-2 + \sqrt
{-2\,q+4}$.  It would therefore converge to some fixed point $v_\bullet$ of
$F_q$ satisfying $0 < -2 + \sqrt {-2\,q+4} \leq v_\bullet < v_{c,{\rm
PM-FM}}(q)$, which is impossible because $v_{c,{\rm
PM-FM}}(q)$ is the unique positive fixed point of $F_q(v)$.  Therefore, for any $v_0 \in I_3$ there is
some iterate $m_0$ for which $F_q^{m_0}(v_0) \in I_2$, at which point the
reasoning in the previous paragraph implies that $F_q^m(v_0) \rightarrow 0$.

\vspace{0.1in}
We will now prove Claim (ii).  Note that pole $v = -q/2$ cannot occur on $[0,v_{c,{\rm PM-FM}}(q)]$; see Eqn.\ (\ref{EQN:ORDER}) for
the case that $q < 0$.  Therefore, since $F_q'(0) = 0$ and $F_q'(v_{c,{\rm PM-FM}}(q)) > 1$, by Lemma~\ref{LEM:REPELLING_FP}, we have that $F_q(v) > v$ for all $v > v_{c,{\rm PM-FM}}(q)$.
(When $q < 0$ we allow for the possibility that $v=-q/2$ is the pole of $F_q(v)$.)   This implies that if $v_0 > v_{c,{\rm PM-FM}}(q)$ the sequence
of iterates $F_q^m(v_0)$ is increasing.  It must converge to infinity as there is no fixed point of $F_q(v)$ that
is larger than $v_{c,{\rm PM-FM}}(q)$.
\end{proof}

We are now ready to prove Theorems \ref{THM:AFM} and \ref{THM:FM}.  We will begin with Theorem \ref{THM:FM} as the proof is somewhat easier.

\subsection{Proof of Theorem \ref{THM:FM}.}
Recall from Section \ref{SUBSEC:RG_FP} that the fixed points of $F_q(v)$ other
than $0$ and $\infty$ are roots of the equation
\begin{align*}
-v^3+q^2+2qv = 0.
\end{align*}
Solving for the values of $q$ for which $F_q(v)$ has a fixed point at $v_0$
yields the two solutions
\begin{align*}
q_\pm(v_0) = \left( -1 \pm \sqrt {1+v_0} \right) v_0.
\end{align*}
Since we consider $v_0 > 0$, the resulting fixed point will necessarily be
$v_{c,{\rm PM-FM}}(q)$, as it is the only positive fixed point of $F_q(v)$.

If $q < q_-(v_0)$ or $q > q_+(v_0)$ then $0 < v_0 < v_{c,{\rm PM-FM}}(q)$ and
$F_q^m(v_0) \rightarrow 0$ by Lemma \ref{LEM:FM_DYNAMICS}.  Therefore,
\begin{align*}
\left(-\infty, q_-(v_0)\right) \cup \left(q_+(v_0),\infty\right) \in \mathcal{P}^0_q(v_0)
\end{align*}
and Proposition \ref{PROP_DESCRIBING_FIGURES} implies that $B_q(v_0)$ does not
intersect such points.  (Such points are colored white in Figures
\ref{bq_v1_figure} - \ref{bq_v99_figure}.)

When $q = q_\pm(v_0)$ the marked point $a(q) = v_0$ hits the fixed point
$v_{c,{\rm PM-FM}}(q)$, which is repelling by Lemma \ref{LEM:REPELLING_FP}.
Since there are values of $q$ for which $a(q) \neq v_{c,{\rm PM-FM}}(q)$ this
corresponds to the marked point $a(v_0)$ being mapped by $F_q^0(v)$
non-persistently to the repelling fixed point $v_{c,{\rm PM-FM}}(q)$.
Therefore, the parameters $q_\pm(v_0)$ are active parameters, and hence
$q_\pm(v_0) \in B_q(v_0)$, by Lemma \ref{LEM:ACTIVE_SUBSET_BQ}.

Finally, if $q_-(v_0) < q < q_+(v_0)$ then $v_0 > v_{c,{\rm PM-FM}}(q)$ and
$F_q^m(v_0) \rightarrow \infty$ by Lemma \ref{LEM:FM_DYNAMICS}.  Therefore,
\begin{align*}
\left(q_-(v_0),q_+(v_0)\right) \in \mathcal{P}^\infty_q(v_0)
\end{align*}
and Proposition \ref{PROP_DESCRIBING_FIGURES} implies that $B_q(v_0)$ does not intersect such points.
(Such points are colored blue in Figures \ref{bq_v1_figure} - \ref{bq_v99_figure}.)
\qed

\subsection{Proof of Theorem \ref{THM:AFM}.}

Let $-1 \leq v_0 < 0$. Lemma \ref{LEM:PART_OF_AFM_DYNAMICS}(i) gives for any $q
< 0$ that $F_q^m(v_0) \rightarrow 0$ and hence $(-\infty,0) \in
\mathcal{P}_q^0(v_0)$.  Proposition \ref{PROP_DESCRIBING_FIGURES} implies that
$B_q(v_0)$ does not intersect such points.

We will now show that $q = 0$ is in $B_q(v_0)$.  This requires some care
because the degree of $F_q(v)$ drops from $4$ to $2$ at $q=0$.  (It is why we
omitted $q=0$ from the parameter space in the discussion from Section
\ref{complex_dynamics}.)  For this reason, Lemma \ref{LEM:ACTIVE_SUBSET_BQ}
does not immediately apply to $q=0$.  Instead, we will show for any $\epsilon >
0$ that there is a parameter $q \neq 0$ with $|q| < \epsilon$ that is active
for $a(q) \equiv v_0$ under~$F_q^m$.  Such a point is in $B_q(v_0)$ by Lemma
\ref{LEM:ACTIVE_SUBSET_BQ}, and, since $\epsilon > 0$ is arbitrary, this will
prove that $0 \in B_q(v_0)$.

As seen in the previous paragraph, any real point $q < 0$ is in
$\mathcal{P}_q^0(v_0)$.  Note that as $q \searrow 0$ the fixed point
$v_-(q)$ increases to $0$; see Figure \ref{FIG_FIXED_POINTS}.  Therefore, we
can choose $q > 0$ sufficiently small so that $-1 \leq v_0 < v_-(q)$.  It then
follows from Lemma \ref{LEM:PART_OF_AFM_DYNAMICS}(ii) that $F_q^m(v_0)
\rightarrow v_-(q)$.  We conclude that on any arbitrarily small annulus $0 <
|q| < \epsilon$ there are two different passive behaviors for the marked point
$a(q) \equiv v_0$:  (i) convergence to $0$ and (ii) convergence to $v_-(q) \neq
0$.  Therefore, there must be active parameters in this annulus.

We will now prove that $q_c(v_0)=\left( -2-\sqrt {-v_0} \right) v_0$ is the largest point where $B_q(v_0)$ hits the real $q$-axis.
One finds that
\begin{align*}
\frac{\partial F_q(v_0)}{\partial q} = -2\,{\frac {{v_0}^{2} \left( {v_0}^{2}+q+2\,v_0 \right) }{ \left( q+2\,v_0 \right) ^{3}}},
\end{align*}
which is negative for $q > -2v_0$.  We have
\begin{align*}
\lim_{q \rightarrow -2v_0} F_q(v_0) = \infty \qquad \mbox{and} \qquad \lim_{q \rightarrow \infty} F_q(v_0) = 0.
\end{align*}
Since $v_{c,{\rm PM-FM}}(q)$ is positive and increasing for $q > 0$ the Intermediate Value Theorem gives that
for any $-1 \leq v_0 < 0$ there is a unique $q_c(v_0) > -2v_0$ for which 
\begin{align*}
F_{q_c(v_0)}(v_0) = v_{c,{\rm PM-FM}}(q).
\end{align*}
In particular, when $q=q_c(v_0)$ the marked point $a(q) \equiv v_0$ lands
non-persistently on the repelling fixed point $v_{c,{\rm PM-FM}}(q)$ when
mapped by $F_q(v)$.  This implies that $q_c(v_0)$ is an active parameter for
$a(q) \equiv v_0$ and hence that $q_c(v_0) \in B_q(v_0)$, by Lemma
\ref{LEM:ACTIVE_SUBSET_BQ}.

Moreover, $0 < F_{q_c(v_0)}(v_0) < v_{c,{\rm PM-FM}}(q)$ for $q > q_c(v_0)$.
Lemma \ref{LEM:FM_DYNAMICS}(i) then implies that such $q$ are in
$\mathcal{P}^0_q(v_0)$ and Proposition \ref{LEM:BQ_NOT_SUBSET_P} then implies that they are not in
$B_q(v_0)$.

The formula $q_c(v_0) = \left( -2-\sqrt {-v_0} \right) v_0$ is obtained by solving $F_q^2(v_0) = F_q(v_0)$ for $q$ and selecting
the branch of the solutions that is larger than the pole $q=-2v_0$.

\vspace{0.1in}

It remains to show that there is a sequence of real parameters
$q_k(v_0) \in B_q(v_0)$ that converge to $q_\infty(v_0)$; See Eq.
(\ref{EQN:QINFINITY}).  

\vspace{0.1in}
\noindent
{\bf Case 1:} $-1 \leq v_0 \leq -\frac{8}{9}$:

\vspace{0.1in}
\noindent
When $q = \frac{32}{27}$ a direct calculation shows that $F_q^m(v_0) \leq -\frac{8}{9} = v_-(q) = v_+(q)$ for
all $m \geq 0$ and $F_q^m(v_0) \rightarrow -\frac{8}{9}$ as $m \rightarrow
\infty$. When $q=3$ we have
\begin{align*}
F_3(v) = {\frac {{v}^{2} \left( {v}^{2}+4\,v+6 \right) }{ \left( 3+2\,v
 \right) ^{2}}} \geq 0 \qquad \mbox{for all} \qquad v \in \mathbb{R}.
\end{align*}
We claim that for any $k \geq 2$ the function $q \mapsto F^k_q(v_0)$ has at
least one pole in the interval $\left(\frac{32}{27},3\right]$.  To see this,
note that if $F^{k-1}_q(v_0)$ has a pole $q_\bullet \in
\left(\frac{32}{27},3\right]$, then $q_\bullet$ is also a pole of
$F^{k}_q(v_0)$ since $F_q(\infty) = \infty$ for any $q$.  Otherwise, if
$F^{k-1}_q(v_0)$ has no pole in the interval $\left(\frac{32}{27},3\right]$
then
\begin{align*}
F_{\frac{32}{27}}^{k-1}(v_0) \leq -\frac{8}{9} < -\frac{q}{2}  \qquad \mbox{and} \qquad F_3^{k-1}(v_0) \geq 0 > -\frac{q}{2}
\end{align*}
so that the Intermediate Value Theorem implies that there is some $q_\bullet$ with $F^{k-1}_{q_\bullet}(v_0) = -\frac{q}{2}$.
Since $-\frac{q}{2}$ is a pole of $F_q(v_0)$ this implies that $F_{q_\bullet}^{k}(v_0) = \infty$.

For any $k \geq 2$ 
let $q_{k}^\infty(v_0)$ be the smallest pole of $F_q^k(v_0)$ in the interval
$\left(\frac{32}{27},3\right]$.  
We have
\begin{align*}
F_{32/27}^k(v_0) \leq  -\frac{8}{9} \qquad \mbox{and}  \qquad \lim_{q \rightarrow q_{k}^\infty(v_0)}F_q^k(v_0) \rightarrow +\infty,
\end{align*}
with the fact that the limit is $+\infty$ (not $-\infty$) a consequence of Lemma \ref{LEM:INVARIANT}.
For any $q \in
\left[\frac{32}{27},3\right]$ the repelling fixed point $v_{c,{\rm PM-FM}}(q)$
remains in $[0,3]$. 
Therefore, the Intermediate Value Theorem implies that there exists $q_k(v_0) \in \left(\frac{32}{27},q_{k}^\infty(v_0)\right)$ with 
\begin{align*}
F_{q_k(v_0)}^k(v_0) = v_{c,{\rm PM-FM}}(q_k(v_0)).
\end{align*}
As this does not hold identically for all $q \in
\left(\frac{32}{27},q_{k}^\infty(v_0)\right)$, we conclude that the marked point
$a(q)~\equiv~v_0$ is mapped by $F_q^k$ non-persistently at $q_k(v_0)$.  It
follows from  Lemma \ref{LEM:ACTIVE_SUBSET_BQ} that $q_k(v_0)~\in~{\mathcal B}_q(v_0)$.

For any $q_0 > \frac{32}{27}$ there is a definite constant $C(q_0)$ such that 
\begin{align*}
F_q(v) > v + C(q_0) \qquad \mbox{for all $q > q_0$ and all $-1 \leq v < -\frac{q}{2}$.}
\end{align*}
Meanwhile, $F_q(v) > 0$ for any $v > -\frac{q}{2}$.  Therefore,
there is a definite $K(q_0)$ such that for any $-1 \leq v_0 < 0$ we have $F_q^{K(q)}(v_0) \geq 0$.  Since $[0,\infty)$ is 
invariant under $F_q$ for $q > 0$ and the pole $-\frac{q}{2}$ is negative, this
implies that for any $k > K(q)$ we have
\begin{align*}
\frac{32}{27} < q_{k}^\infty(v_0) < q_0.
\end{align*}
Since $q_0 > \frac{32}{27}$ was arbitrary, we conclude that $q_k(v_0)
\rightarrow \frac{32}{27}$ as $k \rightarrow \infty$.

\vspace{0.1in}
\noindent
{\bf Case 2:} $-\frac{8}{9} < v_0 < 0$:

\vspace{0.1in}
\noindent
A straightforward modification of the proof from Case 1 applies, after replacing $q = \frac{32}{27}$ with the unique value of $q = \left( -1-\sqrt {1+v_0} \right) v_0$ for which
$v_-(q) = v_0$.

\qed


\section{Zeros in the $v$ Plane }
\label{yzeros}

In addition to the zeros of $Z(D_m,q,v)$ in the $q$ plane for fixed $v=v_0$ and
their accumulation locus ${\mathcal B}_q(v_0)$ in the limit $m \to \infty$, it is
also of interest to investigate the zeros of $Z(D_m,q,v)$ in the $v$ plane for
fixed $q=q_0$ and their accumulation locus ${\mathcal B}_v(q_0)$ in this limit $m
\to \infty$.  We present some new results on these in this section.  

As was noted in
\cite{ddi}, these zeros form the the Julia set of the transformation $F_q(v)$
\cite{juliaset}.  This can be understood as follows.  Assume that, for a fixed
$q$, $v_0'$ is a zero of $Z(D_m,q,v)$, i.e., $Z(G_m,q,v_0')=0$. Then from
Eqs.\ (\ref{zdhl})-(\ref{vp}), it follows that
\beq
Z(D_{m+1},q,v_0) = Z(D_m,q,v_0')(q+2v_0)^{2 \cdot 4^m} \ . 
\label{zdhlv0}
\eeq
Hence, $Z(D_{m+1},q,v_0)=0$ if $v_0'=v_0$. This connects the set of points
$v_0'$ for which $Z(D_m,q,v_0)=0$ to the set of values of $v$ that are left
invariant by the transformation $F_q(v)$.  Since a deviation of a given $v$
from this invariant set will generically be multiplied by successive iterations
of $F_q(v)$, this yields the identification of ${\mathcal B}_v$ with the Julia set
of $F_q(v)$. Previous studies have presented zeros in the $y$ plane for low
values of $q$.

Here we extend this study of zeros of $Z(D_m,q,v)$ in $v$ with new
results for large positive and negative values of $q$. Although
negative values of $q$ do not have direct significance for the
physical $q$-state Potts model, they are of interest in investigating
the mathematical properties of the full $Z(D_m,q,v)$ polynomial.  
In \cite{lq} we showed that ${\mathcal B}_v$ crosses the positive $v$
axis at
\beq
v \simeq |q|^{2/\Delta_{eff}} \ .
\label{yc_largeq}
\eeq
Since, as noted above, $\Delta_{eff}$ is close to 3 for $D_4$, it follows that
the zeros to cross the positive $v$ axis at $v \simeq q^{2/3}$ for $|q| \gg 1$.
In agreement with this, we show our calculations of zeros of $Z(D_4,q,v)$ in
the $y$ plane for the large positive values $q=10^2$, $10^3$, and the large
negative value $q=-10^2$ in Figs. 
\ref{dhlyplotq100}-\ref{dhlyplotq1000}. (The number of zeros displayed in
each of these figures, namely 256, is sufficiently great that the plotting
program yields what appears to be a curve, although it is really a discrete
set of zeros.)

There is a striking change in the structure of the zeros as $q$ grows
to values that are $\gg 1$ \cite{ddi,bambihu89,yangzeng}. In particular, the 
Hausdorff dimensionality of the Julia set approaches 1 from above 
asymptotically as $q \to \infty$ \cite{bambihu89,gao11}. 
Related to this, the global pattern of zeros becomes a single Jordan curve 
rather than the more complicated patterns observed for
small $q$.  The fact that this Jordan curve is not too different from a circle
is similar to the pattern of complex-temperature zeros for the Potts
model on regular lattices in the limit of large $q$ \cite{hungwu,lq}.


\begin{figure}
\begin{center}
\scalebox{1.2}{
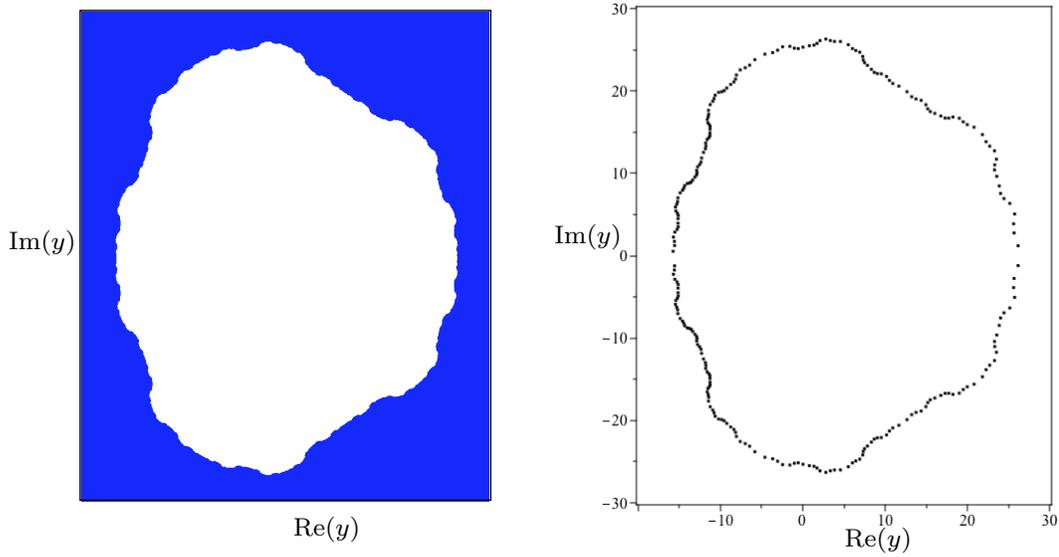
}
\end{center}
\caption{Left: Region diagram for $D_\infty$ in the complex $y=v+1$ plane for $q=100$.  
    The locus ${\mathcal B}_v(100)$ is comprised of the
    boundary between the white and blue.  (See text for further discussion.)
Right: Zeros of the partition function $Z(D_4,100,y)$ (256 zeros).
    Both left and right figures depict
    $-20 < {\rm Re}(y) < 30$ and $-30~<~{\rm Im}(q)~<~30$. }
\label{dhlyplotq100}
\end{figure}
%

\begin{figure}
\begin{center}
\scalebox{1.2}{
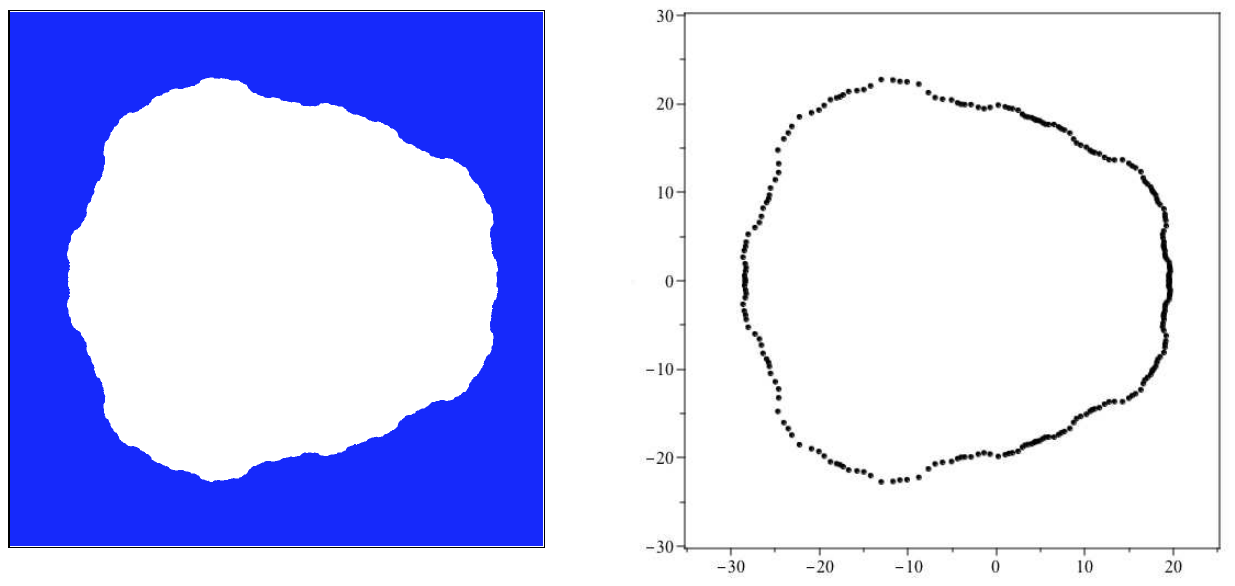
}
\end{center}
\caption{Left: Region diagram for $D_\infty$ in the complex $y=v+1$ plane for $q=-100$.  
    The locus ${\mathcal B}_v(-100)$ is comprised of the
    boundary between the white and blue.  (See text for further discussion.)
Right: Zeros of the partition function $Z(D_4,-100,y)$ (256 zeros).
    Both left and right figures depict
    $-35 < {\rm Re}(y) < 25$ and $-30~<~{\rm Im}(q)~<~30$. }
\label{dhlyplotqminus100}
\end{figure}
%

\begin{figure}
\begin{center}
\scalebox{1.2}{
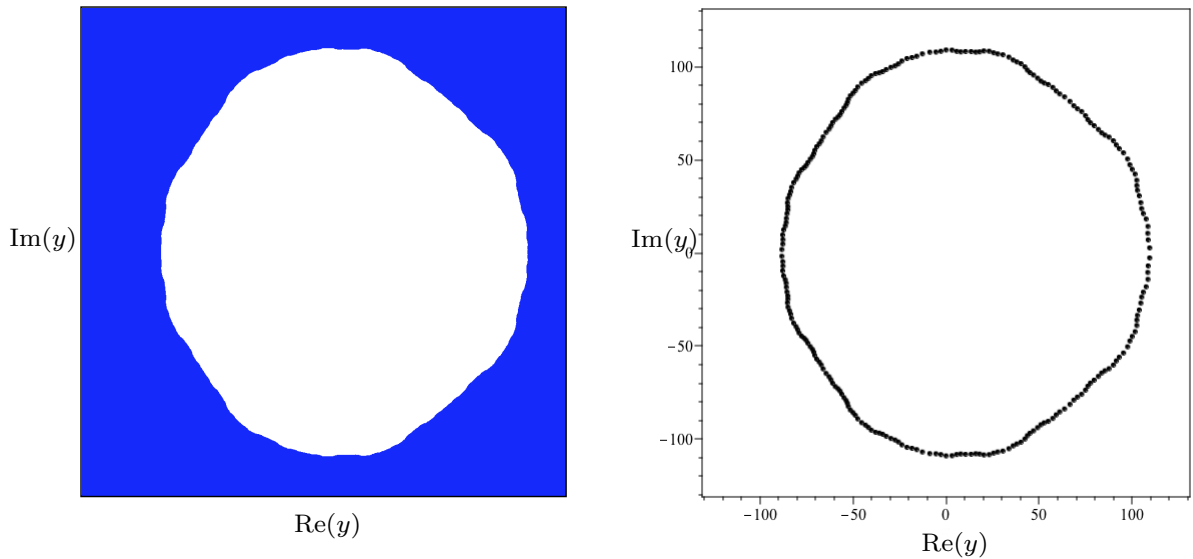
}
\end{center}
\caption{Left: Region diagram for $D_\infty$ in the complex $y=v+1$ plane for $q=1000$.
    The locus ${\mathcal B}_v(1000)$ is comprised of the
    boundary between the white and blue.  (See text for further discussion.)
Right: Zeros of the partition function $Z(D_4,1000,y)$ (256 zeros).
    Both left and right figures depict
    $-130 < {\rm Re}(y) < 130$ and $-130 < {\rm Im}(q) < 130$. }
\label{dhlyplotq1000}
\end{figure}
%


\section{Conclusions }
\label{conclusions}

In this paper we have derived several exact results on the continuous
accumulation sets of zeros ${\mathcal B}_q(v_0)$ of the Potts model on the
Diamond Hierarchical Lattice, $D_\infty$, in the complex $q$-plane at various
fixed values of a temperature-like Boltzmann variable $v = v_0$.  We have
applied methods from complex dynamics to determine the region diagram in which
the infinite iteration of the renormalization group transformation exhibits
different behavior in the various regions and related them to the locus
${\mathcal B}_q(v_0)$.  We have also used these techniques to prove rigorous
results (Theorems~\ref{THM:AFM} and \ref{THM:FM}) about the intersection of
${\mathcal B}_q(v_0)$ with the real $q$-axis.  For the chromatic zeros (i.e.,
partition function zeros of the zero-temperature Potts antiferromagnet), we
have shown that the locus ${\mathcal B}_q(v_0)$ crosses the real $q$ axis at
$q=0$, $q=3$, $q=32/27$, and an infinite set of points between 32/27 and the
value $q=q_1 \simeq 1.639$ given analytically in Eq. (\ref{q1}). 
A similar behavior occurs for any finite-temperature Potts antiferromagnet on
the DHL.  For the finite-temperature Potts ferromagnet on the DHL, the
locus ${\mathcal B}_q(v_0)$ crosses the real $q$ axis at only two points. We
have also studied region diagrams and the structure of ${\mathcal B}_q(v_0)$
for the finite-temperature Potts antiferromagnet and for the Potts ferromagnet
on the DHL and compared them with the patterns of zeros calculated for $D_m$
with $m=4$.  Another result of our present work is that as $|v| \to \infty$,
${\mathcal B}_q$ approaches a circular form with $|q| \sim |v|^{3/2}$.
Finally, we determine properties of the locus ${\mathcal B}_v$ for $q \to \pm
\infty$.


\noindent
{\bf Acknowledgments:}
This research was partly supported by the Taiwan Ministry of Science
and Technology grant MOST 103-2918-I-006-016 (S.-C.C.), by
U.S. National Science Foundation 
grant No. NSF-DMS-1348589 (R.R.), and by the U.S. National Science Foundation
grants No. NSF-PHY-1620628 and NSF-PHY-1915093 (R.S.).


\begin{appendix}


\section{On the Case $q=0$}
\label{q0_appendix}

Besides the very delicate potential issues about relating ${\mathcal B}_q(v_0)$ with ${\mathcal B}_v(q_0)$
that we discussed in Section \ref{SUBSEC:C2}, there is an additional problem that occurs when $q=0$, which we explore it in this appendix.
It is related to the more general fact
that in defining the free energy from the partition function, 
it is necessary to take into account a
noncommutativity in the limits $n \to \infty$ and $q \to q_s$ at special values
$q=q_s$ (which may include $q_s=0,1...,\chi(G)$, where $\chi(G_m)$ denotes the
chromatic number of $G_m$).  This was discussed in \cite{w} for the chromatic
polynomial and in \cite{a} for the full free energy.  The noncommutativity is  
\beq
\lim_{n \to \infty} \lim_{q \to q_s} Z(G_m,q,v)^{1/n} 
\ne \lim_{q \to q_s} \lim_{n \to \infty} Z(G_m,q,v)^{1/n} \ .
\label{fnoncomm}
\eeq
As in the text, we denote $G_\infty$ as the formal limit of an iterative 
family of $n$-vertex graphs $G_m$ as $m \to \infty$ and hence $n \to \infty$. 
Because of this noncommutativity, the definitions of both the free energy 
in (\ref{fz}) and ${\mathcal B}_v(q)$ require that the order of limits be specified. 

A simple example will illustrate this.  Consider the Potts model on the 
$n$-vertex circuit graph,~$C_n$.  An elementary calculation yields, for
the partition function, the result
\beq
Z(C_n,q,v) = (q+v)^n + (q-1)v^n.
\label{zcn}
\eeq
As a special case of Eq.\ (\ref{cluster}), $Z(C_n,q,v)$ has $q$ as a factor, so
$Z(C_n,0,v)=0$. As an explicit example of the
noncommutativity (\ref{fnoncomm}) with $q_s=0$, 
\beq
\lim_{n \to \infty} \lim_{q \to 0} Z(G_m,q,v)^{1/n} = 0
\label{fnoncomm_q0}
\eeq
while, in contrast, if one takes the limit $n \to \infty$ first, before taking
$q \to 0$, then 
\beq
\lim_{n \to \infty} Z(G_m,q,v)^{1/n} = 
\begin{cases}
q+v & \mbox{if $|q+v| \ge |v|$} \\ 
v   & \mbox{if $|v|  \ge |q+v|$  \ .}
\end{cases}
\label{lim_ninf}
\eeq
Then finally taking $q \to 0$, one obtains
\beq
\lim_{q \to 0} \lim_{n \to \infty} Z(G_m,q,v)^{1/n} = v \ . 
\label{fnoncomm_ninf}
\eeq
  The loci ${\mathcal B}_v(q)$ and ${\mathcal
  B}_q(v)$ are given by the solution of the condition of equality of dominant
terms in $Z$, namely
\beq
|q+v|=|v| \ .
\label{qvv}
\eeq
In the $q$ plane, the solution locus, ${\mathcal B}_q(v)$, is a circle with radius
$|v|$. If $v$ is real, this circle is centered at $q=-v$ and crosses the real
$q$ axis at the two points $q=-2v$ and $q=0$. Thus, the maximum (finite) point
at which this locus ${\mathcal B}_q(v)$ crosses the real $q$ axis is 
\beq
q_c(C_\infty) =
\begin{cases}
-2v & \mbox{if $v \le 0$} \\ 
0  & \mbox{if $v\ge 0$  \ .}
\end{cases}
\label{qcc}
\eeq
For example, in the case $v=-1$ where $Z(C_n,q,-1)=P(C_n,q)$ is the chromatic
polynomial and ${\mathcal B}_q(-1)$ is the continuous accumulation set of the
chromatic zeros, the locus ${\mathcal B}_q(-1)$ is the unit circle $|q-1|=1$,
which crosses the real $q$ axis at $q=2$ and $q=0$. 

To show the connection with ${\mathcal B}_v(q)$ or equivalently, ${\mathcal B}_y(q)$,
with $v=y-1$, we first note that the solution of Eq.\ (\ref{qvv}) in the complex
$v$ plane is the infinite vertical line crossing the real $v$ axis at $v=-q/2$,
i.e., 
\beq
{\mathcal B}_v(q): \quad v = -\frac{q}{2} + i \lambda, \quad \lambda \in {\mathbb R}.
\label{bvcn}
\eeq
Equivalently, in the complex $y$ plane, ${\mathcal B}_y(q)$ is the infinite
vertical line crossing the real $y$ axis at $y=1-(q/2)$:
\beq
{\mathcal B}_y(q): \ y = 1-\frac{q}{2} + i \lambda, \quad \lambda \in {\mathbb R}.
\label{bycn}
\eeq
It is convenient to change variables to a variable in terms of which the locus
is compact. We choose $\eta=y^{-1}=e^{-K}$; in the complex plane of the variable $\eta$,
if $q \ne 2$, then the locus ${\mathcal B}_\eta (q)$ is a circle:
\beq
{\mathcal B}_\eta(q): \ \eta = \frac{1}{q-2}(-1+e^{i\omega}), \quad 0 \le \omega <
2\pi.
\label{byinv}
\eeq
As is evident from Eq.\ (\ref{byinv}), this circle crosses the real $\eta$
axis at $\eta=0$ and $\eta=-2/(q-2)$.  If $q=2$, then ${\mathcal B}_y(2)$ is the
line $y=i\lambda$ with $\lambda \in {\mathbb R}$, i.e., the imaginary $y$
axis. This locus is invariant under $y \to 1/y$, so the locus 
${\mathcal B}_\eta(q)$ is also the imaginary axis in the $\eta$ plane. 

Here we observe that the point $q=0$ is in the locus ${\mathcal B}_q(v)$
for all $v$, in particular, for $v=-1$, i.e., $y=0$, but $y=0$ is not
in the locus ${\mathcal B}_y(q)$ for $q=0$. Explicitly, if $q=0$, then
${\mathcal B}_y(q=0)$ is the vertical line in the $y$ plane crossing the
real $y$ axis at $y=1$; equivalently, ${\mathcal B}_\eta(q=0)$ is the
circle $|\eta-(1/2)|=1/2$ crossing the real $\eta$ axis at $\eta=0$
and $\eta=1$.

\end{appendix}



\begin{thebibliography}{99}
\bibitem{juliaset}
%
We recall some definitions \cite{beardon,milnor}. Let $F_q(v)$ be an iterated
transformation on $v \in {\mathbb C}$ depending on a parameter $q \in {\mathbb
C}$, i.e.\ $v_{m+1} = F_q(v_m)$.    The Fatou set of $F_q$ is the maximal open
set on which the family of iterates $v \mapsto F_q^m(v)$ form a normal family in the sense
of Montel's Theorem.  
The Julia set of $F_q(v)$, denoted $J(F_q)$, is the
complement of the Fatou set.  Informally, the Fatou set is the set
of points in the $v$ plane for which $F_q^m(v)$ is ``stable,'' i.e.\ those values of $v$ for which the orbits of other initial conditions
chosen sufficiently close to $v$ have similar behavior to the orbit of $v$.  For
example, those initial conditions $v$ for which $F_q^m(v)$ converges
to an attracting periodic orbit are in the Fatou set.  The Julia
set is the set of points in the $v$ plane for which $F_q^m(v)$ is ``unstable''.
Any repelling periodic point of $F_q$ is in $J(F_q)$ and such points
are actually dense in $J(F_q)$.
\end{thebibliography}
\end{document}